\documentclass{tlp}
\def\isdraft{0}

\usepackage{mathptmx}
\usepackage{mathtools}
\usepackage{amsmath,amsfonts,amssymb}
\usepackage[dvipsnames]{xcolor}
\usepackage{bbm}
\usepackage{graphicx}
\usepackage{prettyref} 
\usepackage[utf8]{inputenc}
\usepackage{enumitem}
\usepackage[hidelinks]{hyperref}
\usepackage{todonotes}
\usepackage{stmaryrd}
\graphicspath{{Figures/}}
\usepackage{footnote}
\usepackage[normalem]{ulem}

\newtheorem{theorem}{Theorem}

\newtheorem{corollary}[theorem]{Corollary}
\newtheorem{fact}[theorem]{Fact}
\newtheorem{lemma}[theorem]{Lemma}
\newtheorem{proposition}[theorem]{Proposition}

\newtheorem{definition}[theorem]{Definition}

\newtheorem{example}[theorem]{Example}

\newtheorem{observation}[theorem]{Observation}

\newtheorem{remark}[theorem]{Remark}

\newrefformat{sec}{§\ref{#1}}
\newrefformat{a}{Answer \ref{#1}}
\newrefformat{ana}{Analogy \ref{#1}}
\newrefformat{alg}{Algorithm \ref{#1}}
\newrefformat{ax}{Appendix \ref{#1}}
\newrefformat{cl}{Claim \ref{#1}}
\newrefformat{con}{Conclusion \ref{#1}}
\newrefformat{conv}{Convention \ref{#1}}
\newrefformat{cj}{Conjecture \ref{#1}}
\newrefformat{c}{Corollary \ref{#1}}
\newrefformat{q}{(\ref{#1})}
\newrefformat{e}{Example \ref{#1}}
\newrefformat{exe}{Exercise \ref{#1}}
\newrefformat{d}{Definition \ref{#1}}
\newrefformat{Done}{Done \ref{#1}}
\newrefformat{f}{Fact \ref{#1}}
\newrefformat{fig}{Figure \ref{#1}}
\newrefformat{h}{Hypothesis \ref{#1}}
\newrefformat{idea}{Idea \ref{#1}}
\newrefformat{i}{Item \ref{#1}}
\newrefformat{l}{Lemma \ref{#1}}
\newrefformat{note}{Note \ref{#1}}
\newrefformat{n}{Notation \ref{#1}}
\newrefformat{o}{Observation \ref{#1}}
\newrefformat{pr}{Problem \ref{#1}}
\newrefformat{p}{Proposition \ref{#1}}
\newrefformat{pc}{Pseudocode \ref{#1}}
\newrefformat{Q}{Q. \ref{#1}}
\newrefformat{r}{Remark \ref{#1}}
\newrefformat{ta}{Table \ref{#1}}
\newrefformat{t}{Theorem \ref{#1}}
\newrefformat{w}{Warning \ref{#1}}

\newcommand{\boxblacktriangle}{\mathrel{\ooalign{$\square$\cr\kern0.07ex\hbox{\scalebox{0.9}{$\blacktriangle$}}}}}
\newcommand{\boxtriangle}{\mathrel{\ooalign{$\square$\cr\kern0.07ex\hbox{\scalebox{0.9}{$\triangle$}}}}}

\AtBeginDocument{\renewcommand{\models}{\mathrel|\joinrel=}}

\newcommand{\Triangle}{\mathrel{\triangle}}

\if\isdraft1
\fi

\begin{document}

\lefttitle{%
	Set-like operations on propositional logic programs
}
\righttitle{%
	C. Anti\'c
}

\jnlPage{1}{00}
\jnlDoiYr{2026}
\jyear{2026}
\doival{}

\title{
	Set-like operations on propositional logic programs\footnote{To appear in Theory and Practice of Logic Programming}
}

\begin{authgrp}
\author{%
	CHRISTIAN ANTI\'C
}
\affiliation{%
	Vienna University of Technology, Vienna, Austria\\
	\email{christian.antic@icloud.com}
}
\end{authgrp}

\maketitle

\begin{abstract}
	Composition and decomposition of logic programs have been studied extensively in the context of modularity, and decomposing a program along the dependency structure of its atoms --- by strata or strongly connected components in datalog, and by splitting sets in the non-monotonic setting of answer set programming --- is standard practice. All of these approaches operate on the level of rule sets: programs are cut along dependencies between atoms, while the internal structure of the individual rules remains untouched. In this paper, we complement this picture by a finer-grained algebra. We introduce set-like operations on (propositional Horn) logic programs --- body-union, body-intersection, body-complement, body-subtraction, body-symmetric-difference, and body-power-set --- which manipulate rule bodies in analogy to the corresponding set operations, and we study their algebraic laws and their interaction with sequential composition and the least model semantics. Our main technical result is a decomposition theorem in this algebra: every minimalist program --- containing at most one rule for each rule head --- is the body-union of Krom programs consisting only of rules with at most one body atom, in such a way that its least model is the intersection of the least models of these components; for arbitrary programs we obtain corresponding approximations. Since Krom programs are algebraically better behaved than arbitrary programs --- composition is associative and distributes from the left --- this may support decomposition-based reasoning and provides a basis for compositional program construction.
\end{abstract}


\begin{keywords}
	Algebraic Logic Programming, Logic Program Decomposition, Propositional Horn Theories, Program Synthesis, Program Transformation
\end{keywords}

\section{Introduction}

Rule-based reasoning is an important component of artificial intelligence. The definite propositional logic programs considered in this paper can be viewed as Horn theories equipped with least-model semantics \citep{Apt90,Apt97,Hodges94,Lloyd87,Sterling94}.

Composition and decomposition of logic programs are not new concerns. Standard approaches decompose programs along the dependency structure of their atoms, for example by stratification or strongly connected components in datalog \citep{Ceri89,Ceri90}, and by splitting sets in the non-monotonic setting of answer set programming \citep{Lifschitz94}. These approaches operate primarily at the level of rules or program modules, while the internal structure of individual rule bodies remains unchanged. The present work complements this perspective by studying finer-grained operations that manipulate and decompose rule bodies. To our knowledge, existing algebraic approaches do not directly provide the body-level decompositions and corresponding semantic relationships developed here.

The propositional Horn setting studied in this paper is fundamental to logic programming. Grounding transforms first-order programs into propositional ones, possibly over an infinite alphabet, while function-free programs over finite domains lead to finite propositional programs of the kind considered here. Propositional rule representations obtained by grounding are central to answer set programming \citep{Gelfond91,Marek99} (see also \cite{Baral03,Brewka11,Lifschitz19}) and to datalog \citep{Ceri89,Ceri90}. Closely related structures are studied as propositional Horn theories in knowledge representation and reasoning.

The idea of introducing algebraic structure into programming dates back at least to the work of John Backus \citep{Backus78}, who advocated algebraic methods for constructing and reasoning about programs. In logic programming, algebraic and compositional perspectives have been developed extensively in connection with modularity and program composition \citep{OKeefe85,Mancarella88,Hirani90,Bossi96,Brogi90,Brogi92,Brogi92a,Brogi95,Brogi94,Bugliesi94,Dong88,Dong90,Gaifman89,Hill94,Ioannidis91}. The present approach is complementary to this work in that it manipulates the bodies of individual rules directly. The set-like operations introduced below enrich the algebra based on sequential composition \citep{Antic21-1}, which is recalled in \prettyref{sec:Composition}. This paper also forms part of a broader line of work on algebraic methods for logic program synthesis and transformation \citep{Antic14,Antic21-2,Antic21-1,Antic23-23}.

We introduce a family of set-like operations on propositional Horn logic programs and study their algebraic laws, their interaction with sequential composition, and their relationship to least-model semantics. A rule has the form $a\leftarrow B$, where $a$ is an atom from a finite alphabet $A$ and $B\subseteq A$ is its body. The body may be empty, in which case the rule is a fact. For rules with the same head, the operations introduced in this paper are modeled on the familiar set operations $\cup$, $\cap$, ${}^c$, $-$, $\Triangle$, and $\mathfrak P$:
\begin{align*}
	(a\leftarrow B)\sqcup(a\leftarrow C) &:= a\leftarrow(B\cup C),\\
	(a\leftarrow B)\sqcap(a\leftarrow C) &:= a\leftarrow(B\cap C),\\
	(a\leftarrow B)^\copyright &:= a\leftarrow(A - B),\\
	(a\leftarrow B)\boxminus(a\leftarrow C) &:= a\leftarrow(B - C),\\
	(a\leftarrow B)\boxtriangle(a\leftarrow C) &:= a\leftarrow(B\Triangle C),\\
	(a\leftarrow B)^{\mathfrak p} &:= \{a\leftarrow C\mid C\subseteq B\}.
\end{align*}
These operations are extended to programs by combining rules with the same head. We refer to them as body-union, body-intersection, body-complement, body-subtraction, body-symmetric-difference, and body-power-set, respectively.

Although these operations are modeled on classical set operations, their behavior at the program level is more subtle. In particular, if a program contains several rules with the same head, body-union and body-intersection need not be idempotent, and familiar distributivity laws may hold only as inclusions, modulo subsumption equivalence, or under additional minimality assumptions. We determine a number of such laws and study how the operations interact with reducts, sequential composition, supported models, and least models. We also show that body-union induces a well-defined operation on subsumption-equivalence classes and yields, together with sequential composition, a seminearring structure on the resulting quotient.

The main technical contribution concerns decomposition into Krom programs, that is, programs whose rules contain at most one body atom. Let $M$ be a minimalist program, meaning that it contains at most one rule for each rule head, and let $m=\max\{1,w(M)\}$, where $w(M)$ is the maximal number of body atoms occurring in a rule of $M$. The Minimalist Decomposition Theorem shows that there are Krom programs $K_1,\ldots,K_m\subseteq M^\lor$ such that
\begin{align*}
	M &= K_1\sqcup\ldots\sqcup K_m
\end{align*}
and
\begin{align*}
	M^\omega &= K_1^\omega\sqcup\ldots\sqcup K_m^\omega = K_1^\omega\cap\ldots\cap K_m^\omega,
\end{align*}
where ${}^\omega$ denotes the least-model operator and $M^\lor$ is the Krom program obtained by replacing each proper rule by its one-atom body components. Thus, both the program and its least model admit compatible representations in terms of Krom components. This is a structural existence and representation result rather than, by itself, an independent algorithm for computing $M^\omega$.

For an arbitrary program $P$, rules with the same head may interact under body-union and produce additional rules. Exact syntactic and semantic reconstruction can therefore no longer be guaranteed. Nevertheless, the Decomposition Theorem provides Krom programs $K_1,\ldots,K_m\subseteq P^\lor$, where $m=\max\{1,w(P)\}$, satisfying
\begin{align*}
	P &\subseteq K_1\sqcup\ldots\sqcup K_m
\end{align*}
and
\begin{align*}
	P^\omega &\subseteq K_1^\omega\sqcup\ldots\sqcup K_m^\omega = K_1^\omega\cap\ldots\cap K_m^\omega.
\end{align*}
The distinction between exact decomposition for minimalist programs and approximation for arbitrary programs reflects the additional combinations created by rules sharing the same head.

Krom programs provide particularly simple components because sequential composition is associative when the left factor is Krom and distributes from the left over program union, whereas these properties fail for arbitrary programs. The decomposition results therefore express complex programs in terms of components with better algebraic behavior and may support further work on compositional reasoning, program transformation, and program construction.

In the final section, we sketch possible applications of the proposed operations to automatic logic programming \citep{Deville94}, including inductive logic programming, learning from interpretation transitions, and analogical logic programming. Further work may investigate extensions to first-order and higher-order programs and to non-monotonic programs under stable-model and answer-set semantics.

\section{Preliminaries}

In this section, we recall some basic algebraic structures \prettyref{sec:Alg} occurring in the paper, and the syntax and semantics of (propositional Horn) logic programs \prettyref{sec:LPs}.

\subsection{Algebraic structures}\label{sec:Alg}

A \emph{permutation} of a set $A$ is any mapping $A\to A$ which is one-to-one and onto. 

We denote the \emph{composition} of two functions $f:A\to A$ and $g:A\to A$ by $f\circ g$ with the usual definition $(f\circ g)(x)=f(g(x))$.

We call a binary relation $\leq$ \emph{reflexive} if $x\leq x$, \emph{anti-symmetric} if $x\leq y$ and $y\leq x$ implies $x=y$, and \emph{transitive} if $x\leq y$ and $y\leq z$ implies $x\leq z$, for all $x,y,z$. A \emph{partially ordered set} (or \emph{poset}) is a set $L$ together with a reflexive, transitive, and anti-symmetric binary relation $\leq$ on $L$. A \emph{prefixed point} of an operator $f$ on a poset $L$ is any element $x\in L$ such that $f(x)\leq x$; moreover, we call any $x\in L$ a \emph{fixed point} of $f$ if $f(x)=x$.


We now turn to algebras consisting of a set together with a binary operation \citep{Howie03}. A \emph{magma} is a set $M$ together with a binary operation $\cdot$ on $M$. We call $(M,\cdot,1)$ a \emph{unital magma} if it contains a \emph{unit element} 1 such that $1a=a1=a$ holds for all $a\in M$. A \emph{semigroup} is a magma $(S,\cdot)$ in which $\cdot$ is associative. A \emph{monoid} is a semigroup containing a unit element. 

A \emph{left} (resp., \emph{right}) \emph{zero} is an element $0$ such that $0a=0$ (resp., $a0=0$) holds for all $a$. 

A relation $\sim$ on the set $S$ is called \emph{left} (resp., \emph{right}) \emph{compatible} with the operation on $S$ if $x\sim y$ implies $zx\sim zy$ (resp., $xz\sim yz$), for all $x,y,z\in S$. It is called \emph{compatible} if $x\sim y$ implies $wxz\sim wyz$, for all $w,x,y,z\in S$. A left (resp., right) compatible relation is called a \emph{left} (resp., \emph{right}) \emph{congruence}. A compatible relation is called a \emph{congruence}. A relation on a semigroup $S$ is a congruence if it is both a left and a right congruence. Given an equivalence relation $\sim$, we denote the equivalence class of an element $a$ with respect to $\sim$ by $[a]_\sim$.

A \emph{bimonoid} \citep{Droste10} is a structure $(M,+,\cdot,0,1)$ such that $(M,+,0)$ and $(M,\cdot,1)$ are both monoids. A bimonoid is a \emph{strong bimonoid} if the operation $+$ is commutative and $0$ acts as multiplicative zero, that is, $a\cdot 0 = 0\cdot a = 0$ for every $a\in M$.

A \emph{seminearring} \citep{Golan99,vanHoorn67} is a set $S$ together with two binary operations $+$ and $\cdot$ on $S$, and a constant $0\in S$, such that $(S,+,0)$ is a monoid and $(S,\cdot)$ is a semigroup satisfying the following laws, for all $a,b,c\in S$:
\begin{align} 
	(a+b)\cdot c &= (a\cdot c) + (b\cdot c)\quad\text{(right-distributivity)}\\
	\label{eq:0a=0} 0\cdot a &= 0\quad\text{(left zero).}
\end{align} We say that $S$ is \emph{idempotent} if $a+a=a$ holds for all $a\in S$. Typical examples of seminearrings are given by the set of all mappings on a monoid together with composition of mappings, point-wise addition of mappings, and the zero function. 

A \emph{semiring} is a seminearring $(S,+,\cdot,0)$ such that $+$ is commutative and additionally to the laws of a seminearring the following laws are satisfied, for all $a,b,c\in S$:
\begin{align*} 
	a\cdot (b+c) &= (a\cdot b) + (a\cdot c)\quad\text{(left-distributivity)}\\
	a\cdot 0 &= 0\quad\text{(right zero)}.
\end{align*} For example, $(\mathbb N,+,\cdot,0)$ and $(2^A,\cup,\cap,\emptyset)$ are semirings. A characteristic feature of semirings (as opposed to rings) is that addition needs not be invertible.

\subsection{Logic programs}\label{sec:LPs}

We shall now recall the syntax \prettyref{sec:Syntax} and least model semantics \prettyref{sec:Semantics} of (propositional Horn) logic programs \citep{Apt90} and their sequential composition \prettyref{sec:Composition}.

\subsubsection{Syntax}\label{sec:Syntax}

In the rest of the paper, $A$ denotes a finite alphabet of propositional atoms. 

A (\emph{propositional Horn logic}) \emph{program} over $A$ is a finite set of \emph{rules} of the form
\begin{align}
    \label{eq:r} a\leftarrow b_1,\ldots,b_k,\quad k\geq 0,
\end{align} where $a,b_1,\ldots,b_k\in A$ are propositional atoms, and we denote the set of all programs over $A$ by $\mathbb P_A$. With a slight abuse of notation, we will not distinguish between $a \leftarrow b_1,\ldots,b_k$ and $a \leftarrow \{b_1,\ldots,b_k\}$, which will allow us to write the rule as $a \leftarrow B$ with $B = \{b_1,\ldots,b_k\}$, a notation heavily used in the rest of the paper.

The \emph{width} of a rule $r$ is the cardinality of its body,
\begin{align*}
w(r) := |b(r)|.
\end{align*} The width of a program $P$ is defined by
\begin{align*}
w(P) := \max\bigl(\{w(r)\mid r\in P\}\cup\{0\}\bigr).
\end{align*}

It will be convenient to define, for a rule $r$ of the form \prettyref{eq:r} the \emph{head} and \emph{body operators} $h$ and $b$ as
\begin{align*} 
     h(r) := \{a\} \quad\text{and}\quad  b(r) := \{b_1,\ldots,b_k\}
\end{align*} extended to programs rule-wise.

For any programs $P$ and $R$, we have
\begin{align} 
	\label{eq:h(P_cup_R)} h(P\cup R) &= h(P)\cup h(R),\\
	\label{eq:b(P_cup_R)} b(P\cup R) &= b(P)\cup b(R),
\end{align} which means that the head and body operators are homomorphisms with respect to set union.

A \emph{fact} is a rule with empty body and a \emph{proper rule} is a rule which is not a fact. We denote the facts and proper rules in $P$ by $ f(P)$ and $ p(P)$, respectively. We have
\begin{align*} 
	f(P\cup R) &= f(P)\cup f(R),\\
	p(P\cup R) &= p(P)\cup p(R),
\end{align*} which means that the facts and proper rules operators are homomorphisms with respect to set union as well.

We call a rule $r$ of the form \prettyref{eq:r} \emph{Krom}\footnote{Krom rules were first introduced and studied in \cite{Krom67} and are sometimes called ``binary'' in the literature.} if it has at most one body atom, and we call a program \emph{Krom} if it contains only Krom rules. We denote the set of all Krom programs over $A$ by $\mathbb K_A$.

A \emph{minimalist} program contains at most one rule for each rule head and we denote the set of all minimalist programs over $A$ by $\mathbb M_A$. For example, the program
\begin{align*} 
	&a \leftarrow b\\
	&a \leftarrow c
\end{align*} implicitly contains a disjunction as it is logically equivalent to the single rule
\begin{align*} 
	a \leftarrow (b\lor c).
\end{align*} Hence, a program is minimalist iff it does not implicitly contain disjunctions of that form.


The \emph{unit program} over $A$ is the minimalist Krom program
\begin{align} 
    \label{eq:1} 1_A := \{a\leftarrow a\mid a\in A\}.
\end{align}

\subsubsection{Semantics}\label{sec:Semantics}

An \emph{interpretation} over $A$ is any set of atoms from $A$ and we denote the set of all such interpretations by $\mathbb I_A$ or simply by $\mathbb I$ in case $A$ is understood. 

We define the \emph{entailment relation}, for every interpretation $I$, inductively as follows: (i) for an atom $a$, $I\models a$ if $a\in I$; (ii) for a set of atoms $B$, $I\models B$ if $B\subseteq I$; (iii) for a rule $r$ of the form \prettyref{eq:r}, $I\models r$ if $I\models  b(r)$ implies $I\models  h(r)$; and, finally, (iv) for a propositional logic program $P$, $I\models P$ if $I\models r$ holds for each rule $r\in P$. In case $I\models P$, we call $I$ a \emph{model} of $P$ and we denote the set of all models of $P$ by $Mod(P)$. The set of all models of $P$ has a least element with respect to set inclusion called the \emph{least model} of $P$ and denoted by $LM(P)$. 

Define the \emph{van Emden-Kowalski operator} \citep{vanEmden76} $T_P : \mathbb I_A\to \mathbb I_A$ of $P$, for every interpretation $I\in \mathbb I_A$, by
\begin{align*} 
    T_P(I) := \{a\in A \mid a \leftarrow B\in P,\, B\subseteq I\}.
\end{align*}

The van Emden-Kowalski operator is at the core of logic programming since the models of $P$ coincide with the prefixed points of $T_P$ with respect to the inclusion ordering thus providing an operational characterization of models:

\begin{theorem}[\cite{vanEmden76}]\label{t:Mod(P)} An interpretation $I$ is a model of a program $P$ iff $T_P(I)\subseteq I$, that is,
\begin{align*} 
	Mod(P) = \{I\in \mathbb I_A\mid T_P(I)\subseteq I\}.
\end{align*}
\end{theorem}

We call an interpretation $I$ a \emph{supported model} of $P$ if $I$ is a fixed point of $T_P$ and we denote the set of all supported models of $P$ by $Supp(P)$.

We say that $P$ and $R$ are \emph{subsumption equivalent} \citep{Maher88} --- denoted by $P\approx R$ --- if $T_P = T_R$. Notice that subsumption equivalence implies that $P$ and $R$ have the same models.

The van Emden-Kowalski operator is monotone on the set of all interpretations ordered by set inclusion and its \emph{least fixed point computation} is given by
\begin{align*} 
    &T_P^0 = \emptyset,\\
    &T_P^{n+1} = T_P(T_P^n),\\
    &T_P^\infty = \bigcup_{n\geq 0}T_P^n.
\end{align*} The following constructive characterization of least models is due to \cite{vanEmden76}:

\begin{proposition}\label{prop:LM} The least model of a program coincides with the least fixed point of its associated van Emden-Kowalski operator, that is, for any program $P$ we have
\begin{align} 
    LM(P) = T_P^\infty.
\end{align}
\end{proposition}

\subsubsection{Sequential composition}\label{sec:Composition}

We define the \emph{(sequential) composition} \citep{Antic21-1} of two programs $P$ and $R$ as
\begin{align*} 
	P\circ R := \{a \leftarrow (B_1\cup\ldots\cup B_k) \mid a \leftarrow b_1,\ldots,b_k\in P,\; b_1 \leftarrow B_1,\ldots,b_k \leftarrow B_k\in R,\, k\geq 0\},
\end{align*} which we will usually denote simply by $PR$. 

For example, we have
\begin{align*} 
	\left\{
	\begin{array}{l}
		a\\
		b \leftarrow a\\
		c \leftarrow a,b
	\end{array}
	\right\}\left\{
	\begin{array}{l}
		a\\
		b \leftarrow a\\
	\end{array}
	\right\} = \left\{
	\begin{array}{l}
		a\\
		b\\
		c \leftarrow a
	\end{array}
	\right\}.
\end{align*}

Composition satisfies the following identities (among others):
\begin{align} 
	P1 &= 1P = P\\
	\label{eq:(P_cup_R)Q} (P\cup R)Q &= PQ\cup RQ\\
	\label{eq:P0=f(P)} P\emptyset &= f(P)\\
	\label{eq:h(P)=PA} h(P) &= PA.
\end{align} 

Unfortunately, composition is \textit{not} associative \cite[Example 10]{Antic21-1} and does not distribute from the left over union since
\begin{align} 
	\label{eq:P(Q_cup_R)}\{a \leftarrow b,c\}(\{b\}\cup \{c\}) = \{a\} \quad\text{whereas}\quad \{a \leftarrow b,c\}\{b\}\cup \{a \leftarrow b,c\}\{c\} = \emptyset.
\end{align} 

However, for minimalist programs $M$ and $N$, composition \textit{is} associative in the sense that
\begin{align} 
	\label{eq:(PM)N} (PM)N = P(MN)
\end{align} holds for every program $P$ \cite[Theorem 11]{Antic21-1}. The same holds true for a Krom program $K$ and all programs $P$ and $R$ \cite[Theorem 12]{Antic21-1}, that is,
\begin{align} 
	\label{eq:KPR} K(PR) = (KP)R.
\end{align} Moreover, in contrast to the general case (see \prettyref{eq:P(Q_cup_R)} above), Krom programs distribute from the left thus satisfying \cite[Theorem 12]{Antic21-1}
\begin{align} 
	\label{eq:K(P_cup_R)} K(P\cup R) &= KP\cup KR.
\end{align}

Together with the corresponding closure properties, this yields the following structural result:

\begin{proposition} The structure $( \mathbb K_A, \cup, \circ, \emptyset, 1_A )$ forms a seminearring with unit $1_A$ and the structure $( \mathbb M_A, \circ, 1_A )$ forms a monoid.
\end{proposition}
\begin{proof} See \citep[Theorems 11 and 12]{Antic21-1}.
\end{proof}


The van Emden-Kowalski operator $T_P$ can be represented via composition in the sense that
\begin{align} 
	\label{eq:T_P(I)=PI} T_P(I) = PI.
\end{align} This correspondence allows us to directly define concepts in terms of programs instead of the van Emden-Kowalski operator: 

\begin{fact}\label{f:PI=I} An interpretation $I$ is a model of a program $P$ iff $PI\subseteq I$. Moreover, $I$ is a supported model of $P$ iff $PI=I$.
\end{fact}

\begin{fact} Two programs $P$ and $R$ are subsumption equivalent iff $PI=RI$ holds for every interpretation $I$.
\end{fact}
 
More specifically, we first define
\begin{align*} 
	P\lesssim R \quad:\Leftrightarrow\quad PI\subseteq RI,\quad\text{for every interpretation $I$},
\end{align*} and
\begin{align*} 
	P\approx R \quad:\Leftrightarrow\quad P\lesssim R \quad\text{and}\quad R\lesssim P.
\end{align*} We clearly have
\begin{align} 
	\label{eq:P_leqq_ss_R_} P\lesssim R \quad\Rightarrow\quad LM(P)\subseteq LM(R).
\end{align}

Subsumption equivalence is a congruence with respect to sequential composition and set union as a direct consequence of \cite[Corollary 36]{Antic21-1} which means that we can form the factor algebra
\begin{align} 
	\label{eq:mathfrak_P} [\mathfrak P_A] := \left([\mathbb P_A], \cup, \circ, [\emptyset], [1_A]\right),
\end{align} where
\begin{align*} 
	[\mathbb P_A] := \{[P] \mid P\in \mathbb P_A\},
\end{align*} and
\begin{align*} 
	[P]\cup [R] := [P\cup R] \quad\text{and}\quad [P]\circ [R] := [P\circ R],
\end{align*} isomorphic to
\begin{align*} 
	\mathfrak T_A := (\mathbb T_{ \mathbb P_A}, \cup, \circ, T_ \emptyset, T_{1_A}),
\end{align*} where
\begin{align*} 
	\mathbb T_{ \mathbb P_A} := \{T_P \mid P\in \mathbb P_A\}.
\end{align*} denotes the set of all van Emden-Kowalski operators over $A$.

We have the following structural result:

\begin{proposition}\label{p:seminearring} The isomorphic structures $\mathfrak T_A$ and $[\mathfrak P_A]$ are seminearrings.
\end{proposition}

We now turn to the least model semantics of programs. For this, define the \emph{(Kleene) star} of a program $P$ by
\begin{align*} 
	P^\ast := \bigcup_{n\geq 0}P^n,
\end{align*} where (recall that composition is not associative)
\begin{align*} 
	P^0 := 1, \qquad P^1 := P, \qquad P^{n+1} := P^n P.
\end{align*} 

Moreover, define the \emph{omega} of $P$ by
\begin{align} 
	\label{eq:P^omega} P^\omega := f(P^\ast) \stackrel{\ref{eq:P0=f(P)}}= P^\ast\emptyset.
\end{align}

For every interpretation $I$, we have
\begin{align} 
	\label{eq:I^ast} I^\ast &= 1\cup I\\
	I^\omega &= I.
\end{align}

We have the following algebraic characterization of the least model of a program as shown in \cite[Theorem 40]{Antic21-1}: the omega of $P$ coincides with the least model of $P$, that is,
\begin{align} 
	\label{eq:LM(P)=P^omega} P^\omega = LM(P).
\end{align}

For an interpretation $I$, define the minimalist program
\begin{align*} 
    I^\oplus := \{a\leftarrow(\{a\}\cup I)\mid a\in A\}.
\end{align*} The program
\begin{align*} 
    A^\oplus &= \{a\leftarrow A\mid a\in A\}
\end{align*} will play an important role in this paper as it occurs as the neutral element with respect to body-intersection (see \prettyref{f:sqcap_monoid}).

\subsubsection{Reducts}

\begin{definition} Given an interpretation $I$, we define the \emph{head-reduct} and \emph{body-reduct} of a program $P$ with respect to $I$ respectively by\footnote{The body-reduct coincides with the Faber-Leone-Pfeifer reduct \citep{Faber11} which is well-known in answer set programming.}
\begin{align*} 
	^IP := \{a \leftarrow B\in P \mid a\in I\} \quad\text{and}\quad P^I := \{a \leftarrow B\in P \mid B\subseteq I\}.
\end{align*}
\end{definition}

We now recollect some basic properties of reducts (see \cite[§4.3]{Antic21-1}). 

For two interpretations $I$ and $J$, we have
\begin{align} 
	\label{eq:I^J} I^J &= I\\
	\label{eq:^JI} {^J}I &= I\cap J.
\end{align}

According to \cite[Proposition 17]{Antic21-1}, the body-reduct operator is an endomorphism on $(\mathbb P_A,\cup,\cap,\emptyset,A)$, that is, for any programs $P,R$ and interpretation $I$, we have
\begin{align} 
	(P\cup R)^I &= P^I\cup R^I\\
	(P\cap R)^I &= P^I\cap R^I\\
	\emptyset^I &= \emptyset\\
	A^I &= A,
\end{align} and the head-reduct operator is an endomorphism on $(\mathbb P_A,\cup,\cap,\emptyset)$, that is, 
\begin{align} 
	{^I}(P\cup R) &= {^I}P\cup {^I}R\\
	{^I}(P\cap R) &= {^I}P\cap {^I}R\\
	{^I}\emptyset &= \emptyset.
\end{align} Moreover, the left and right reducts can be rewritten in terms of composition as
\begin{align}
	^IP &= (1^I)P\\
	P^I &= P1^I
\end{align} and
\begin{align} 
	\label{eq:(PR)^I} (PR)^I &= PR^I\\
	^I(PR) &= {^I}PR.
\end{align}

\section{Body-union}\label{sec:Body-union}

The operations introduced in this section and the following ones form a natural family of set-like transformations on rule bodies. While body-union plays a central role in the decomposition results of \prettyref{sec:Decomposition}, the remaining operations are studied to obtain a more complete picture of the algebraic structure of logic programs and to understand how these operations interact with program semantics.

This section introduces the first set-like operation on logic programs in analogy to the union of sets:

\begin{definition} We define the \emph{body-union} of $P$ and $R$ by
\begin{align*} 
	P\sqcup R &:= \{a \leftarrow (B\cup C)\mid a \leftarrow B\in P,\; a \leftarrow C\in R\}.
\end{align*}
\end{definition}

For example,
\begin{align*} 
	\left\{
	\begin{array}{l}
		a\\
		b \leftarrow a\\
		b \leftarrow b\\
		c \leftarrow a,b
	\end{array}
	\right\} \sqcup \left\{
	\begin{array}{l}
		b\\
		b \leftarrow c
	\end{array}
	\right\} = \left\{
	\begin{array}{l}
		b \leftarrow a\\
		b \leftarrow b\\
		b \leftarrow a,c\\
		b \leftarrow b,c
	\end{array}
	\right\}.
\end{align*} 

We have the following interactions between body-union and some of our basic operations:
\begin{align}
	\label{eq:h(P_sqcup_R)} h(P\sqcup R) &= h(P)\cap h(R)\\
	\label{eq:f(P_sqcup_R)} f(P\sqcup R) &= f(P)\cap f(R)\\
	\label{eq:p(P_sqcup_R)} p(P\sqcup R) &= (f(P)\sqcup p(R))\cup (p(P)\sqcup f(R))\cup (p(P)\sqcup p(R)).
\end{align} Notice that the first two identities mean that $h$ and $f$ are both homomorphisms from $( \mathbb P, \sqcup)$ to $( \mathbb P, \cap)$.


Recall that the union of sets is idempotent and despite the similarity between body-union and union, we have the following observation (but see \prettyref{c:P_sqcup_P_equiv_ss_P} where we will show that body-union is at least idempotent modulo subsumption equivalence):

\begin{observation} Body-union is not idempotent even for Krom programs, that is, there is a Krom program $K$ such that
\begin{align*} 
	K\sqcup K\neq K.
\end{align*} For example, we have
\begin{align*} 
	\left\{
	\begin{array}{l}
		a \leftarrow b\\
		a \leftarrow c
	\end{array}
	\right\} \sqcup \left\{
	\begin{array}{l}
		a \leftarrow b\\
		a \leftarrow c
	\end{array}
	\right\} &= \left\{
	\begin{array}{l}
		a \leftarrow b\\
		a \leftarrow c\\
		a \leftarrow b,c
	\end{array}
	\right\}
\end{align*} However, notice that body-union is idempotent on interpretations as a consequence of the forthcoming \prettyref{eq:I_sqcup_J}. 
\end{observation}

\begin{fact}\label{f:sqcup_monoid} The structure $(\mathbb P_A,\sqcup,A)$ is a commutative monoid, that is,
\begin{align*} 
	(P\sqcup Q)\sqcup R &= P\sqcup (Q\sqcup R)\\
	P\sqcup A &= A\sqcup P = P\\
	P\sqcup R &= R\sqcup P
\end{align*} holds for all programs $P,Q,R\in \mathbb P_A$.
\end{fact}

We now analyze the interaction between body-union and union. Our first observation is the following structural result:

\begin{theorem}\label{t:cup_sqcup} The structure $(\mathbb P_A,\cup,\sqcup,\emptyset,A)$ is a commutative idempotent semiring, that is, in addition to $(\mathbb P_A,\cup,\emptyset)$ being an idempotent commutative monoid and $(\mathbb P_A,\sqcup,A)$ being a commutative monoid (see \prettyref{f:sqcup_monoid}), we have
\begin{align} 
    \label{eq:P_cup_R_sqcup_Q} (P\cup R)\sqcup Q &= (P\sqcup Q)\cup (R\sqcup Q)\\
    \label{eq:Q_sqcup_P_cup_R} Q\sqcup (P\cup R) &= (Q\sqcup P)\cup (Q\sqcup R)\\
    \label{eq:P_sqcup_0=0} \emptyset\sqcup P &= P\sqcup \emptyset = \emptyset.
\end{align} 
\end{theorem}
\begin{proof} A direct consequence of the rule-wise definition of body-union.
\end{proof}

\begin{remark} In particular, $(\mathbb P_A,\cup,\sqcup,\emptyset,A)$ is a strong bimonoid in the sense of \cite{Droste10}. Notice that the structure with reversed roles $(\mathbb P_A,\sqcup,\cup,A,\emptyset)$ is a bimonoid but not a strong one, since $P\cup A\neq A$ in general.
\end{remark}

The next result shows that $\cup$ does in general \textit{not} distribute over $\sqcup$ (but see \prettyref{t:DT_cup_ss} where we show that it distributes modulo subsumption equivalence):

\begin{observation} There are an interpretation $I$ and Krom programs $K$ and $L$ such that
\begin{align*} 
	(I\sqcup K)\cup L\neq (I\cup L)\sqcup (K\cup L).
\end{align*} For example, for
\begin{align*} 
	I &:= \{a\} \quad\text{and}\quad K := \{a \leftarrow b\} \quad\text{and}\quad L := \{a \leftarrow c\}
\end{align*} we obtain
\begin{align*} 
	(\{a\}\sqcup \{a \leftarrow b\})\cup \{a \leftarrow c\} 
		&= \left\{
		\begin{array}{l}
			a \leftarrow b\\
			a \leftarrow c
		\end{array}
		\right\}\\
		&\neq \left\{
		\begin{array}{l}
			a \leftarrow b\\
			a \leftarrow c\\
			a \leftarrow b,c
		\end{array}
		\right\} = (\{a\}\cup \{a \leftarrow c\})\sqcup (\{a \leftarrow b\}\cup \{a \leftarrow c\}).
\end{align*}
\end{observation}

\begin{theorem}\label{t:DT_cup_ss} For any programs $P,Q,R$, we have
\begin{align*} 
	(P\sqcup R)\cup Q \approx (P\cup Q)\sqcup (R\cup Q).
\end{align*}
\end{theorem}
\begin{proof} For every interpretation $I$, we have (the forward references are non-circular and thus harmless)
\begin{align*} 
	[(P\sqcup R)\cup Q]I & \stackrel{\ref{eq:(P_cup_R)Q}}= (P\sqcup R)I\cup QI\\
		& \stackrel{\ref{eq:(P_sqcup_R)Q=PQ_sqcup_RQ}}= (PI\sqcup RI)\cup QI\\
		& \stackrel{\ref{eq:I_sqcup_J}}= (PI\cap RI)\cup QI\\
		&= (PI\cup QI)\cap (RI\cup QI)\\
		& \stackrel{\ref{eq:(P_cup_R)Q}}= (P\cup Q)I\cap (R\cup Q)I\\
		& \stackrel{\ref{eq:I_sqcup_J}}= (P\cup Q)I\sqcup (R\cup Q)I\\
		& \stackrel{\ref{eq:(P_sqcup_R)Q=PQ_sqcup_RQ}}= [(P\cup Q)\sqcup (R\cup Q)]I.
\end{align*}
\end{proof}

\begin{proposition}\label{p:Q_sqcup_P_cap_R} For any programs $P,Q,R$, we have
\begin{align} 
	\label{eq:Q_sqcup_P_cap_R} Q\sqcup (P\cap R)\subseteq (Q\sqcup P)\cap (Q\sqcup R).
\end{align}
\end{proposition}
\begin{proof} 
\begin{align*} 
	Q\sqcup (P\cap R) &= \{a \leftarrow (B\cup C) \mid a \leftarrow B\in Q,\; a \leftarrow C\in P\cap R\}\\
		&= \{a \leftarrow (B\cup C) \mid a \leftarrow B\in Q,\; a \leftarrow C\in P,\; a \leftarrow C\in R\}\\
		&\subseteq \{a \leftarrow (B\cup C) \mid a \leftarrow B\in Q,\; a \leftarrow C\in P\}\cap \{a \leftarrow (B\cup C)\mid a \leftarrow B\in Q,\; a \leftarrow C\in R\}\\
		&= (Q\sqcup P)\cap (Q\sqcup R),
\end{align*} where the inclusion cannot be turned into an identity since a rule occurring in both programs on the right may originate from a rule in $P$ and a rule in $R$ which are distinct and hence do not occur in $P\cap R$.
\end{proof}

\begin{observation} For $Q := \{a \leftarrow b\}$ and $P := \{a \leftarrow c\}$ and $R := \{a \leftarrow b,c\}$, we have $P\cap R = \emptyset$ and therefore
\begin{align*} 
	Q\sqcup (P\cap R) \stackrel{\ref{eq:P_sqcup_0=0}}= \emptyset \subsetneq \{a \leftarrow b,c\} = \{a \leftarrow b,c\}\cap \{a \leftarrow b,c\} = (Q\sqcup P)\cap (Q\sqcup R).
\end{align*}
\end{observation}

We now study the interplay between body-union and the sequential composition of programs which is used to compute the semantic ${}^\omega$-operator. We first prove that composition distributes from the right over body-union (but not from the left; see \prettyref{o:K(L_sqcup_N)}):

\begin{theorem}\label{t:(P_sqcup_R)Q} For any programs $P,Q,R$, we have
\begin{align} 
	\label{eq:(P_sqcup_R)Q_subseteq} (P\sqcup R)Q &\subseteq PQ\sqcup RQ,\\
	\label{eq:(P_sqcup_R)Q_approx} (P\sqcup R)Q &\approx PQ\sqcup RQ,
\end{align} and equality
\begin{align} 
	\label{eq:(P_sqcup_R)Q=PQ_sqcup_RQ} (P\sqcup R)Q = PQ\sqcup RQ
\end{align} holds in each of the following two cases: (i) $B\cap C = \emptyset$ for all $a \leftarrow B\in P$ and $a \leftarrow C\in R$; (ii) $Q\sqcup Q = Q$.
\end{theorem}
\begin{proof} Given a set of atoms $B$, we call a family $(C_b)_{b\in B}$ of sets of atoms a \emph{$Q$-selection for $B$} if $b \leftarrow C_b\in Q$ holds for every $b\in B$, so that the definition of composition reads as
\begin{align*} 
	PQ = \left\{a \leftarrow \bigcup_{b\in B}C_b \;\middle|\; a \leftarrow B\in P,\; (C_b)_{b\in B}\ \text{a $Q$-selection for $B$}\right\}.
\end{align*} We compute
\begin{align*} 
	(P\sqcup R)Q 
		&= \{a \leftarrow (B\cup C) \mid a \leftarrow B\in P,\; a \leftarrow C\in R\}Q\\
		&= \left\{a \leftarrow \bigcup_{b\in B\cup C}F_b \;\middle|\; a \leftarrow B\in P,\; a \leftarrow C\in R,\; (F_b)_{b\in B\cup C}\ \text{a $Q$-selection for $B\cup C$}\right\}\\
		&\subseteq \left\{a \leftarrow \left(\bigcup_{b\in B}D_b\cup \bigcup_{b\in C}E_b\right) \;\middle|\; a \leftarrow B\in P,\; a \leftarrow C\in R,\; (D_b)_{b\in B},\, (E_b)_{b\in C}\ \text{$Q$-selections}\right\}\\
		&= \left\{a \leftarrow \bigcup_{b\in B}D_b \;\middle|\; a \leftarrow B\in P,\; (D_b)_{b\in B}\ \text{a $Q$-selection for $B$}\right\}\\
			& \qquad\sqcup \left\{a \leftarrow \bigcup_{b\in C}E_b \;\middle|\; a \leftarrow C\in R,\; (E_b)_{b\in C}\ \text{a $Q$-selection for $C$}\right\}\\
		&= PQ\sqcup RQ,
\end{align*} which proves \prettyref{eq:(P_sqcup_R)Q_subseteq}. The inclusion in the third line holds since the restrictions $(F_b)_{b\in B}$ and $(F_b)_{b\in C}$ of a $Q$-selection for $B\cup C$ are $Q$-selections for $B$ and $C$ satisfying $\bigcup_{b\in B}F_b\cup \bigcup_{b\in C}F_b = \bigcup_{b\in B\cup C}F_b$. It is an identity --- and hence proves \prettyref{eq:(P_sqcup_R)Q=PQ_sqcup_RQ} --- in each of the cases (i) and (ii), since given $Q$-selections $(D_b)_{b\in B}$ and $(E_b)_{b\in C}$, the family
\begin{align*} 
	F_b := 
		\begin{cases}
			D_b & b\in B - C,\\
			E_b & b\in C - B,\\
			D_b\cup E_b & b\in B\cap C,
		\end{cases}
\end{align*} is a $Q$-selection for $B\cup C$ satisfying $\bigcup_{b\in B\cup C}F_b = \bigcup_{b\in B}D_b\cup \bigcup_{b\in C}E_b$: in case (i) the third line is vacuous as $B\cap C=\emptyset$, and in case (ii) we have $b \leftarrow (D_b\cup E_b)\in Q\sqcup Q = Q$ for every $b\in B\cap C$.

Finally, we prove \prettyref{eq:(P_sqcup_R)Q_approx}. Notice first that every interpretation $J$ satisfies $J\sqcup J \stackrel{\ref{eq:I_sqcup_I}}= J$, which means that \prettyref{eq:(P_sqcup_R)Q=PQ_sqcup_RQ} is applicable whenever the right factor is an interpretation. Hence, for every interpretation $I$ we compute
\begin{align*} 
	[(P\sqcup R)Q]I 
		& \stackrel{\prettyref{p:seminearring}}{\approx} (P\sqcup R)(QI) \stackrel{\ref{eq:(P_sqcup_R)Q=PQ_sqcup_RQ}}= P(QI)\sqcup R(QI) \stackrel{\ref{eq:I_sqcup_J}}= P(QI)\cap R(QI)\\
		& \stackrel{\prettyref{p:seminearring}}{\approx} (PQ)I\cap (RQ)I \stackrel{\ref{eq:I_sqcup_J}}= (PQ)I\sqcup (RQ)I \stackrel{\ref{eq:(P_sqcup_R)Q=PQ_sqcup_RQ}}= [PQ\sqcup RQ]I,
\end{align*} where we may replace $\approx$ by $=$ since all programs occurring in the computation are interpretations and subsumption equivalent interpretations coincide.
\end{proof}


\begin{remark} Condition (ii) covers the two cases occurring throughout this paper: every interpretation $I$ satisfies $I\sqcup I \stackrel{\ref{eq:I_sqcup_I}}= I$, every minimalist program $M$ satisfies $M\sqcup M = M$, and the program $A^\ast \stackrel{\ref{eq:I^ast}}= 1_A\cup A$ used in \prettyref{p:s_morphism} satisfies $A^\ast\sqcup A^\ast = A^\ast$. On the other hand, $Q := \{b \leftarrow d,\; b \leftarrow e\}$ satisfies $Q\sqcup Q = Q\cup \{b \leftarrow d,e\}\neq Q$, and indeed $P := \{a \leftarrow b\}$ and $R := \{a \leftarrow b,c\}$ together with $Q':=Q\cup \{c \leftarrow c\}$ give
\begin{align*} 
	(P\sqcup R)Q' = \left\{
	\begin{array}{l}
		a \leftarrow c,d\\
		a \leftarrow c,e
	\end{array}
	\right\} \subsetneq \left\{
	\begin{array}{l}
		a \leftarrow c,d\\
		a \leftarrow c,e\\
		a \leftarrow c,d,e
	\end{array}
	\right\} = PQ'\sqcup RQ',
\end{align*} which shows that the inclusion in \prettyref{eq:(P_sqcup_R)Q_subseteq} can in general not be turned into an equality.
\end{remark}

\begin{lemma}\label{l:sqcup_congruence} Subsumption equivalence is a congruence with respect to body-union, that is, $P\approx P'$ and $R\approx R'$ imply $P\sqcup R\approx P'\sqcup R'$, for all programs $P,P',R,R'$.
\end{lemma}
\begin{proof} Since every interpretation $I$ satisfies $I\sqcup I \stackrel{\ref{eq:I_sqcup_I}}= I$, we may apply \prettyref{eq:(P_sqcup_R)Q=PQ_sqcup_RQ} to compute
\begin{align*} 
	(P\sqcup R)I \stackrel{\ref{eq:(P_sqcup_R)Q=PQ_sqcup_RQ}}= PI\sqcup RI = P'I\sqcup R'I \stackrel{\ref{eq:(P_sqcup_R)Q=PQ_sqcup_RQ}}= (P'\sqcup R')I,
\end{align*} where the middle identity holds by assumption.
\end{proof}

\begin{theorem} The structure $( [\mathbb P_A], \sqcup, \circ, [A], [1_A] )$ forms a seminearring with unit $[1_A]$, where $[P]\sqcup [R] := [P\sqcup R]$.
\end{theorem}
\begin{proof} Body-union is well-defined on subsumption equivalence classes by \prettyref{l:sqcup_congruence}, and $( [\mathbb P_A], \sqcup, [A] )$ is a monoid since $( \mathbb P_A, \sqcup, A )$ is one by \prettyref{f:sqcup_monoid}. Composition is associative modulo subsumption equivalence by \prettyref{p:seminearring}, right-distributivity is \prettyref{eq:(P_sqcup_R)Q_approx}, and $[A]$ is a left zero with respect to composition since
\begin{align*} 
	AP = A
\end{align*} holds for every program $P\in \mathbb P_A$ --- indeed, $A$ consists of facts only, which are left unchanged by composition. Finally, $[1_A]$ is a unit since $P1_A = 1_AP = P$.
\end{proof}

\begin{observation}\label{o:K(L_sqcup_N)} Unfortunately, the distributivity law in \prettyref{eq:(P_sqcup_R)Q_approx} does \textit{not} hold if composition is computed from the left as there are Krom programs $K,L,N$ such that
\begin{align*} 
	K(L\sqcup N) \not\approx KL\sqcup KN,
\end{align*} which means that in general we have
\begin{align*} 
	Q(P\sqcup R)\not\approx QP\sqcup QR.
\end{align*} For example, we have
\begin{align*} 
	\left\{
	\begin{array}{l}
		a \leftarrow b\\
		a \leftarrow c
	\end{array}
	\right\}(\{b \leftarrow b\}\sqcup \{c \leftarrow c\}) = \left\{
	\begin{array}{l}
		a \leftarrow b\\
		a \leftarrow c
	\end{array}
	\right\}\emptyset = \emptyset
\end{align*} whereas
\begin{align*} 
	\left\{
	\begin{array}{l}
		a \leftarrow b\\
		a \leftarrow c
	\end{array}
	\right\}\{b \leftarrow b\}\sqcup \left\{
	\begin{array}{l}
		a \leftarrow b\\
		a \leftarrow c
	\end{array}
	\right\}\{c \leftarrow c\} = \{a \leftarrow b\}\sqcup \{a \leftarrow c\} = \{a \leftarrow b,c\}
\end{align*} and the programs $\emptyset$ and $\{a \leftarrow b,c\}$ are, of course, not subsumption equivalent.
\end{observation}

The next result is analogous to $P\cup P = P$:

\begin{corollary}\label{c:P_sqcup_P_equiv_ss_P} Body-union is idempotent modulo subsumption equivalence, that is, for every program $P$,
\begin{align*} 
	P\sqcup P\approx P.
\end{align*}
\end{corollary}
\begin{proof} For every interpretation $I$,
\begin{align*} 
	(P\sqcup P)I \stackrel{\ref{eq:(P_sqcup_R)Q=PQ_sqcup_RQ}}= PI\sqcup PI \stackrel{\ref{eq:I_sqcup_J}}= PI.
\end{align*}
\end{proof}

\begin{theorem} Every model of $P$ is a model of $P\sqcup R$, for any programs $P$ and $R$. That is,
\begin{align*} 
	Mod(P)\subseteq Mod(P\sqcup R).
\end{align*}
\end{theorem}
\begin{proof} By \prettyref{t:Mod(P)}, we know that an interpretation $I$ is a model of $P$ iff $T_P(I)\subseteq I$, which by \prettyref{eq:T_P(I)=PI} is equivalent to $PI\subseteq I$. Now by \prettyref{eq:(P_sqcup_R)Q=PQ_sqcup_RQ}, we have $(P\sqcup R)I=PI\sqcup RI \stackrel{\ref{eq:I_sqcup_J}}= PI\cap RI\subseteq I$, which means that $I$ is a model of $P\sqcup R$.
\end{proof}

\begin{proposition} An interpretation $I$ is a supported model of $P\sqcup R$ iff $PI\cap RI=I$, that is,
\begin{align*} 
	Supp(P\sqcup R) = \{I\in \mathbb I\mid PI\cap RI = I\}.
\end{align*}
\end{proposition}
\begin{proof} By \prettyref{f:PI=I}, an interpretation $I$ is a supported model of $P\sqcup R$ iff $(P\sqcup R)I \stackrel{\ref{eq:(P_sqcup_R)Q=PQ_sqcup_RQ}}= PI\sqcup RI \stackrel{\ref{eq:I_sqcup_J}}= PI\cap RI = I$.
\end{proof}

The next result is analogous to
\begin{align*} 
	P\cup R = P \quad\Leftrightarrow\quad R\subseteq P.
\end{align*}

\begin{proposition}\label{p:P_sqcup_R_equiv_ss_P} For any programs $P$ and $R$,
\begin{align*} 
	P\sqcup R \approx P \quad\Leftrightarrow\quad P\lesssim R.
\end{align*}
\end{proposition}
\begin{proof} $(\Rightarrow)$ For any interpretation $I$, we have
\begin{align*} 
	PI = (P\sqcup R)I = PI\sqcup RI = PI\cap RI\subseteq RI,
\end{align*} which implies $P\lesssim R$. 

$(\Leftarrow)$ From the assumption $PI\subseteq RI$, for every interpretation $I$, we deduce
\begin{align*} 
	(P\sqcup R)I = PI\sqcup RI = PI\cap RI = PI,
\end{align*} which implies $P\sqcup R \approx P$.
\end{proof}

Reducing a program to a restricted alphabet is a fundamental operation on programs which we will often use in the rest of the paper and in this section we analyze the interactions between body-union and reducts.

We can simplify the computation of the body-union using head-reducts as
\begin{align*} 
	P\sqcup R &= {^{h(R)}}P\sqcup {^{h(P)}}R.
\end{align*}

Notice that we have
\begin{align*} 
	{^{f(R)}}P\cup {^{f(P)}}R\subseteq P\sqcup R.
\end{align*}

Computing the body of $P\sqcup R$ is simple using head-reducts:
\begin{align} 
	\label{eq:b(P_sqcup_R)} b(P\sqcup R) &= b({^{h(R)}}P)\cup b({^{h(P)}}R).
\end{align}

Interestingly, the left reduct can be represented via body-union as (compare with \prettyref{eq:P_boxminus_I})
\begin{align}
	\label{eq:I_sqcup_P} I\sqcup P &= P\sqcup I = {^I}P.
\end{align} This implies that for any interpretations $I$ and $J$,
\begin{align} 
	\label{eq:I_sqcup_J} I\sqcup J & \stackrel{\ref{eq:I_sqcup_P}}= {^J}I \stackrel{\ref{eq:^JI}}= I\cap J,
\end{align} and therefore
\begin{align} 
	\label{eq:I_sqcup_I} I\sqcup I = I.
\end{align} Moreover, it implies that computing the proper rules in $P\sqcup R$ can be written as
\begin{align} 
	\label{eq:p(P_sqcup_R)-2} p(P\sqcup R) \stackrel{\ref{eq:p(P_sqcup_R)},\ref{eq:I_sqcup_P}}= (p(P)\sqcup p(R))\cup {^{f(R)}}p(P)\cup {^{f(P)}}p(R).
\end{align}

\begin{proposition}\label{p:body_reduct_sqcup} The body-reduct operator is an endomorphism on the monoid $(\mathbb P_A,\sqcup,A)$, that is,
\begin{align*} 
    (P\sqcup R)^I &= P^I\sqcup R^I\\
    A^I & \stackrel{\ref{eq:I^J}}= A.
\end{align*} Moreover, the head-reduct operator is an endomorphism on the semigroup $(\mathbb P_A,\sqcup)$, that is,
\begin{align} 
	\label{eq:^I_P_sqcup_R} ^I(P\sqcup R) = {^I}P\sqcup {^I}R.
\end{align}
\end{proposition}
\begin{proof} The first two identities hold trivially. To prove the third, we compute
\begin{align*} 
	{^I}(P\sqcup R) & \stackrel{\ref{eq:I_sqcup_P}}= I\sqcup P\sqcup R \stackrel{\ref{eq:I_sqcup_J}}= I\sqcup P\sqcup I\sqcup R \stackrel{\ref{eq:I_sqcup_P}}= {^I}P\sqcup {^I}R.
\end{align*}
\end{proof}

\begin{remark} Notice that the head-reduct operator fails to be an endomorphism on the monoid $(\mathbb P_A,\sqcup,A)$ since, in case $I\neq A$, we have $^IA \stackrel{\ref{eq:^JI}}= I\cap A = I \neq A$.
\end{remark}

\begin{proposition} For any programs $P$ and $R$ and interpretation $I$,
\begin{align} 
	\label{eq:P_cup_I_sqcup_R_cup_I} (P\cup I)\sqcup (R\cup I) = I\cup {^I}P\cup {^I}R\cup (P\sqcup R).
\end{align}
\end{proposition}
\begin{proof} 
\begin{align*} 
	(P\cup I)\sqcup (R\cup I) &= \{a \leftarrow B\cup C\mid a \leftarrow B\in P\cup I,\;a \leftarrow C\in R\cup I\}\\
		&= \{a \mid a\in I\}\cup \{a \leftarrow B\mid a \leftarrow B\in P,\; a\in I\}\\
			&\qquad\cup \{a \leftarrow C\mid a\in I,\; a \leftarrow C\in R\}\cup \{a \leftarrow B\cup C\mid a \leftarrow B\in P,\;a \leftarrow C\in R\}\\
		&= I\cup (I\sqcup P)\cup (I\sqcup R)\cup (P\sqcup R)\\
		& \stackrel{\ref{eq:I_sqcup_P}}= I\cup {^I}P\cup {^I}R\cup (P\sqcup R).
\end{align*}
\end{proof}

\section{Body-intersection}\label{sec:Body-intersection}

We now turn to our next set-like operation which is analogous to set intersection:

\begin{definition} Define the \emph{body-intersection} of $P$ and $R$ by
\begin{align*} 
	P\sqcap R &:= \{a \leftarrow (B\cap C)\mid a \leftarrow B\in P,\; a \leftarrow C\in R\}.
\end{align*}
\end{definition} 

Body-intersection is not idempotent even on Krom programs as is witnessed by the following counterexample:
\begin{align*} 
	\left\{
	\begin{array}{l}
		a \leftarrow b\\
		a \leftarrow c
	\end{array}
	\right\} \sqcap \left\{
	\begin{array}{l}
		a \leftarrow b\\
		a \leftarrow c
	\end{array}
	\right\} &= \left\{
	\begin{array}{l}
		a\\
		a \leftarrow b\\
		a \leftarrow c
	\end{array}
	\right\}.
\end{align*} However, body-intersection coincides with intersection on interpretations,
\begin{align} 
	\label{eq:I_sqcap_J} I\sqcap J = I\cap J,
\end{align} which means that it is idempotent on interpretations,
\begin{align*} 
	I\sqcap I & \stackrel{\ref{eq:I_sqcap_J}}= I\cap I = I.
\end{align*}

\begin{fact}\label{f:sqcap_monoid} The structure $(\mathbb P_A,\sqcap,A^\oplus)$ is a commutative monoid, that is,
\begin{align} 
	(P\sqcap Q)\sqcap R &= P\sqcap (Q\sqcap R)\\
	\label{eq:P_sqcap_A^oplus} P\sqcap A^\oplus &= A^\oplus\sqcap P = P\\
	P\sqcap R &= R\sqcap P.
\end{align}
\end{fact}

We have the following interactions between body-intersection and and the head operator:
\begin{align*} 
	h(P\sqcap R) &= h(P)\cap h(R).
\end{align*} This means that the head operator $h$ is a homomorphism from $( \mathbb P, \sqcap)$ to $( \mathbb P, \cap)$.

Unfortunately, computing $f(P\sqcap R)$ is not as simple as computing $f(P\sqcup R)$ (see \prettyref{eq:f(P_sqcup_R)}) and before we give a formula in \prettyref{eq:f(P_sqcap_R)} below, we first introduce the following auxiliary construct: 

\begin{definition} Define the \emph{bowtie} of $P$ and $R$ by
\begin{align*} 
	P\bowtie R &:= (P\ltimes R)\cup (R\ltimes P).
\end{align*} where
\begin{align*} 
	P\ltimes R &:= \{a \leftarrow B\in P \mid \text{there exists }a \leftarrow C\in R\text{ with } B\cap C = \emptyset\}.
\end{align*}
\end{definition}


We can now compute the facts in $P\sqcap R$ as
\begin{align}
	\label{eq:f(P_sqcap_R)} f(P\sqcap R) = (f(P)\cap f(R))\cup h(P\bowtie R) \stackrel{\ref{eq:f(P_sqcup_R)}}= f(P\sqcup R)\cup h(P\bowtie R).
\end{align} 

\begin{observation} There are Krom programs $K,L,M$ such that
\begin{align*} 
    M(K\sqcap L) \not\approx MK\sqcap ML,
\end{align*} which means that in general we have
\begin{align*} 
	Q(P\sqcap R) \neq QP\sqcap QR.
\end{align*} This is analogous to what we have encountered in \prettyref{o:K(L_sqcup_N)} and a similar counterexample works here as well: we have
\begin{align*} 
	\left\{
	\begin{array}{l}
		a \leftarrow b\\
		a \leftarrow c
	\end{array}
	\right\}(\{b \leftarrow b\}\sqcap \{c \leftarrow c\}) = \emptyset
\end{align*} whereas
\begin{align*} 
	\left\{
	\begin{array}{l}
		a \leftarrow b\\
		a \leftarrow c
	\end{array}
	\right\}\{b \leftarrow b\}\sqcap \left\{
	\begin{array}{l}
		a \leftarrow b\\
		a \leftarrow c
	\end{array}
	\right\}\{c \leftarrow c\} = \{a \leftarrow b\}\sqcap \{a \leftarrow c\} = \{a\}.
\end{align*}

Moreover, there are minimalist programs $M,N$ and an interpretation $I$ such that
\begin{align*} 
	(M\sqcap N)I\neq MI\sqcap NI,
\end{align*} which means that in general we have\footnote{Compare to \prettyref{t:(P_sqcup_R)Q}.}
\begin{align*} 
	(P\sqcap R)Q \not\approx PQ\sqcap RQ.
\end{align*} Compare this with the positive result in \prettyref{t:(P_sqcup_R)Q}! To see why, consider the counterexample
\begin{align*} 
	(\{a \leftarrow b,c\}\sqcap \{a \leftarrow c,d\})\{c\} = \{a\} \neq \{a \leftarrow b,c\}\{c\}\sqcap \{a \leftarrow c,d\}\{c\} = \emptyset.
\end{align*}
\end{observation}




We shall now analyze the interaction between body-intersection and reducts. Our first observation is as follows:

\begin{observation} In contrast to \prettyref{p:body_reduct_sqcup}, the body-reduct operator is in general \textit{not} an endomorphism on the monoid $(\mathbb P_A,\sqcap,A^\oplus)$ since, in general, we have
\begin{align*} 
	(P\sqcap R)^I \not\subseteq P^I\sqcap R^I
\end{align*} as is witnessed by the following counterexample:
\begin{align*} 
	(\{a \leftarrow b\}\sqcap \{a \leftarrow c\})^{\{b\}} = \{a\}^{\{b\}} \stackrel{\ref{eq:I^J}}= \{a\}
\end{align*} whereas
\begin{align*} 
	\{a \leftarrow b\}^{\{b\}}\sqcap \{a \leftarrow c\}^{\{b\}} = \{a \leftarrow b\}\sqcap \emptyset = \emptyset.
\end{align*} However, we do have
\begin{align*}
	P^I\sqcap R^I\subseteq (P\sqcap R)^I
\end{align*} since
\begin{align*} 
	P^I\sqcap R^I = \{a \leftarrow B\cap C\mid B,C\subseteq I\}\subseteq \{a \leftarrow B\cap C\mid B\cap C\subseteq I\} = (P\sqcap R)^I.
\end{align*}
\end{observation}

\begin{proposition} The head-reduct operator is an endomorphism on the semigroup $(\mathbb P_A,\sqcap)$, that is,
\begin{align*} 
	{^I}(P\sqcap R) &= {^I}P\sqcap {^I}R.
\end{align*}
\end{proposition}
\begin{proof} We have
\begin{align*}
	{^I}(P\sqcap R) & \stackrel{\ref{eq:I_sqcup_P}}= I\sqcup (P\sqcap R) \stackrel{\ref{eq:(P_sqcap_R)_sqcup_M}}= (I\sqcup P)\sqcap (I\sqcup R) \stackrel{\ref{eq:I_sqcup_P}}= {^I}P\sqcap {^I}R,
\end{align*} where the forward reference to \prettyref{eq:(P_sqcap_R)_sqcup_M} is non-circular and thus harmless.
\end{proof}

\begin{proposition} We have the following formula for the computation of the head-reduct:
\begin{align} 
	\label{eq:^IP=^IA^oplus_sqcap_P} {^I}P &= {^I}(A^\oplus)\sqcap P.
\end{align}
\end{proposition}
\begin{proof}
\begin{align*} 
	{^I}(A^\oplus)\sqcap P = \{a \leftarrow A \mid a\in I\}\sqcap P = \{a \leftarrow (B\cap A) \mid a \leftarrow B\in P,\, a\in I\} = \{a \leftarrow B\in P\mid a\in I\} = {^I}P.
\end{align*}
\end{proof}

\section{Body-complement}\label{sec:Body-complement}

In analogy to the complement of sets, we introduce the following unary operation on programs:

\begin{definition} We define the \emph{body-complement} of $P$ by
\begin{align*} 
	P^\copyright := \left\{a \leftarrow (A- B) \;\middle|\; a \leftarrow B\in P\right\}.
\end{align*}
\end{definition}

\begin{example} Let $A := \{a,b,c\}$. We then have, for instance,
\begin{align*} 
	\emptyset^\copyright = \emptyset,\qquad \{a\}^\copyright = \{a \leftarrow a,b,c\},\qquad \left\{
	\begin{array}{l}
		a\\
		b \leftarrow a\\
		c \leftarrow a,b
	\end{array}
	\right\}^\copyright = \left\{
	\begin{array}{l}
		a \leftarrow a,b,c\\
		b \leftarrow b,c\\
		c \leftarrow c
	\end{array}
	\right\}.
\end{align*}
\end{example}

Body-complement is involutive,
\begin{align} 
	\label{eq:P^copyright^copyright} (P^\copyright)^\copyright = P.
\end{align} Moreover, it satisfies the following identities (among others):
\begin{align} 
    h(P^\copyright) &= h(P)\\
    b(P^\copyright) &= A- \bigcap_{a \leftarrow B\in P} B\\
    f(P^\copyright) &= h(P\cap A^\oplus) = \{a\in A\mid a \leftarrow A\in P\}\\
    \label{eq:A^copyright=A^oplus} A^\copyright &= A^\oplus\\
    \label{eq:I^copyright} I^\copyright &= {^I}(A^\oplus).
\end{align} 


\begin{fact} Body-complement is an endomorphism on the boolean algebra $(\mathbb P_A,\cup,\cap,{}^c,\emptyset,F_A)$, that is,
\begin{align}
	\label{eq:P_cup_R^copyright} (P\cup R)^\copyright &= P^\copyright\cup R^\copyright\\
    \label{eq:P_cap_R^copyright} (P\cap R)^\copyright &= P^\copyright\cap R^\copyright\\
    (P^c)^\copyright &= (P^\copyright)^c\\
	\emptyset^\copyright &= \emptyset\\ 
	F_A^\copyright &= F_A,
\end{align} where
\begin{align} 
	\label{eq:F_A} F_A := \{a \leftarrow B \mid a\in A ,\, B\subseteq A\}
\end{align} is the \emph{full program} over $A$ consisting of all possible rules and facts that can be formed from $A$.
\end{fact}

This gives us the following formula for the computation of body-complement:
\begin{align*} 
    P^\copyright = (f(P)\cup p(P))^\copyright \stackrel{\ref{eq:P_cup_R^copyright}}= f(P)^\copyright\cup p(P)^\copyright \stackrel{\ref{eq:I^copyright}}= {^{f(P)}}(A^\oplus)\cup p(P)^\copyright.
\end{align*}

The next result shows that we can smoothly adapt De Morgan's laws to body-union, body-intersection, and body-complement as desired:

\begin{proposition} For any programs $P$ and $R$,
\begin{align} 
	\label{eq:(P_sqcup_R)^copyright} (P\sqcup R)^\copyright & = P^\copyright\sqcap R^\copyright\\
	\label{eq:(P_sqcap_R)^copyright} (P\sqcap R)^\copyright & = P^\copyright\sqcup R^\copyright.
\end{align}
\end{proposition}
\begin{proof} Using De Morgan's laws, we compute
\begin{align*} 
	(P\sqcup R)^\copyright &= \{a \leftarrow (B\cup C) \mid a \leftarrow B\in P,\; a \leftarrow C\in R\}^\copyright\\
		&= \{a \leftarrow (A- (B\cup C)) \mid a \leftarrow B\in P,\; a \leftarrow C\in R\}\\
		&= \{a \leftarrow ((A- B)\cap (A- C)) \mid a \leftarrow B\in P,\; a \leftarrow C\in R\}\\
		&= \{a \leftarrow (A- B)\mid  a \leftarrow B\in P\}\sqcap \{a \leftarrow (A- C)\mid a \leftarrow C\in R\}\\
		&= P^\copyright\sqcap R^\copyright.
\end{align*} and a similar computation shows the second identity.
\end{proof}



The body-complement of a Krom program $K$ is
\begin{align*} 
	K^\copyright = \{a \leftarrow A\mid a\in f(K)\}\cup \{a \leftarrow (A-\{b\}) \mid a \leftarrow b\in p(K)\}
\end{align*} which can be written more compactly as
\begin{align*} 
	K^\copyright = {^{f(K)}}(A^\oplus)\cup p(K)^\copyright.
\end{align*}

The identities
\begin{align*} 
	(P^\copyright)^I &= \{a \leftarrow (A-B) \mid a \leftarrow B\in P,\; A- B\subseteq I\}\\
	(P^I)^\copyright &= \{a \leftarrow (A-B) \mid a \leftarrow B\in P,\; B\subseteq I\}
\end{align*} show that in general we have
\begin{align*} 
	(P^\copyright)^I\neq (P^I)^\copyright,
\end{align*} which means that the body-reduct does not commute with body-complement, but see \prettyref{p:^I_copyright} below where we show that the head-reduct does commute with body-complement. At least we (trivially) have
\begin{align*} 
	(P^\copyright)^A = (P^A)^\copyright.
\end{align*}

\begin{proposition}\label{p:^I_copyright} Body-complement commutes with head-reducts, that is, for any program $P$ and interpretation $I$,
\begin{align} 
    \label{eq:^IP^copyright=^IP^copyright} {^I}(P^\copyright) &= ({^I}P)^\copyright.
\end{align}
\end{proposition}
\begin{proof} An immediate consequence from the definition of body-complement and head-reducts. To illustrate some of the identities obtained above, we additionally compute
\begin{align*} 
	{^I}(P^\copyright) & \stackrel{\ref{eq:I_sqcup_P}}= I\sqcup P^\copyright \stackrel{\ref{eq:(P_sqcup_R)^copyright}}= (I^\copyright\sqcap P)^\copyright \stackrel{\ref{eq:I^copyright}}= ({^I}(A^\oplus)\sqcap P)^\copyright \stackrel{\ref{eq:^IP=^IA^oplus_sqcap_P}}= ({^I}P)^\copyright.
\end{align*}
\end{proof}

\begin{theorem}\label{t:sqcup_sqcap} For all programs $P,Q,R$, we have
\begin{align}
	\label{eq:(P_sqcup_R)_sqcap_Q} (P\sqcup R)\sqcap Q &\subseteq (P\sqcap Q)\sqcup (R\sqcap Q)\\
	\label{eq:(P_sqcap_R)_sqcup_Q} (P\sqcap R)\sqcup Q &\subseteq (P\sqcup Q)\sqcap (R\sqcup Q),
\end{align} and in case $M$ is a minimalist program, we have the identities
\begin{align}
	\label{eq:(P_sqcup_R)_sqcap_M} (P\sqcup R)\sqcap M &= (P\sqcap M)\sqcup (R\sqcap M)\\
	\label{eq:(P_sqcap_R)_sqcup_M} (P\sqcap R)\sqcup M &= (P\sqcup M)\sqcap (R\sqcup M).
\end{align}
\end{theorem}
\begin{proof} Applying the well-known distributivity laws of the set operations, we compute
\begin{align*} 
	(P\sqcup R)\sqcap Q &= \{a \leftarrow ((B\cup C)\cap D) \mid a \leftarrow B\in P,\; a \leftarrow C\in R,\; a \leftarrow D\in Q\}\\
		&= \{a \leftarrow ((B\cap D)\cup (C\cap D)) \mid a \leftarrow B\in P,\; a \leftarrow C\in R,\; a \leftarrow D\in Q\}\\
		&\subseteq \{a \leftarrow ((B\cap D)\cup (C\cap D')) \mid a \leftarrow B\in P,\; a \leftarrow C\in R,\; a \leftarrow D\in Q,\; a \leftarrow D'\in Q\}\\
		&= (P\sqcap Q)\sqcup (R\sqcap Q)
\end{align*} and, analogously,
\begin{align*} 
	(P\sqcap R)\sqcup Q &= \{a \leftarrow ((B\cap C)\cup D) \mid a \leftarrow B\in P,\; a \leftarrow C\in R,\; a \leftarrow D\in Q\}\\
		&= \{a \leftarrow ((B\cup D)\cap (C\cup D)) \mid a \leftarrow B\in P,\; a \leftarrow C\in R,\; a \leftarrow D\in Q\}\\
		&\subseteq \{a \leftarrow ((B\cup D)\cap (C\cup D')) \mid a \leftarrow B\in P,\; a \leftarrow C\in R,\; a \leftarrow D\in Q,\; a \leftarrow D'\in Q\}\\
		&= (P\sqcup Q)\sqcap (R\sqcup Q).
\end{align*} In both computations, the inclusion is an identity in case $Q$ is minimalist, since then $D$ and $D'$ are both the body of the unique rule in $Q$ with head $a$ --- which proves \prettyref{eq:(P_sqcup_R)_sqcap_M} and \prettyref{eq:(P_sqcap_R)_sqcup_M}.
\end{proof}

\begin{observation} The inclusions in \prettyref{eq:(P_sqcup_R)_sqcap_Q} and \prettyref{eq:(P_sqcap_R)_sqcup_Q} can in general not be turned into identities. For example, for the non-minimalist program $Q := \{a \leftarrow b,\; a \leftarrow c\}$ and $P := \{a \leftarrow b\}$ and $R := \{a \leftarrow c\}$, we have
\begin{align*} 
	(P\sqcup R)\sqcap Q &= (P\sqcap R)\sqcup Q = \left\{
	\begin{array}{l}
		a \leftarrow b\\
		a \leftarrow c
	\end{array}
	\right\} \subsetneq \left\{
	\begin{array}{l}
		\underline{a}\\
		a \leftarrow b\\
		a \leftarrow c\\
		\underline{a \leftarrow b,c}
	\end{array}
	\right\} = (P\sqcap Q)\sqcup (R\sqcap Q)\\ 
	&= (P\sqcup Q)\sqcap (R\sqcup Q).
\end{align*}
\end{observation}

\section{Body-subtraction}\label{sec:Body-subtraction}

In this section, we introduce an operation on programs which mimics set difference:

\begin{definition} Define the \emph{body-subtraction} of $P$ and $R$ by
\begin{align*} 
	P\boxminus R &:= \{a \leftarrow (B-C)\mid a \leftarrow B\in P,\; a \leftarrow C\in R\}.
\end{align*}
\end{definition} 

For example, we have
\begin{align*} 
	\left\{
	\begin{array}{l}
		a\\
		b \leftarrow a\\
		c \leftarrow a,b
	\end{array}
	\right\}\boxminus \left\{
	\begin{array}{l}
		b \leftarrow a\\
		c \leftarrow a
	\end{array}
	\right\} = \left\{
	\begin{array}{l}
		b\\
		c \leftarrow b
	\end{array}
	\right\}.
\end{align*}

This operation satisfies the following identities (among others):
\begin{align*} 
	P\boxminus\emptyset &= \emptyset\boxminus P = \emptyset\\
	P\boxminus A &= P\\
	P\boxminus A^\oplus &= h(P).
\end{align*}

Moreover, we have
\begin{align*} 
	M\boxminus M &= h(M)
\end{align*} for every minimalist program $M$. Notice that the identity may fail for non-minimalist programs as is witnessed by the counterexample:
\begin{align*} 
	\left\{
	\begin{array}{l}
		a \leftarrow b\\
		a \leftarrow c
	\end{array}
	\right\}\boxminus \left\{
	\begin{array}{l}
		a \leftarrow b\\
		a \leftarrow c
	\end{array}
	\right\} = \left\{
	\begin{array}{l}
		a\\
		a \leftarrow b\\
		a \leftarrow c
	\end{array}
	\right\}.
\end{align*} We thus have
\begin{align*} 
	P\boxminus P\supseteq h(P).
\end{align*}

\begin{fact} The head operator is a homomorphism $h:(\mathbb P_A,\boxminus)\to (\mathbb I_A,\cap)$ thus satisfying
\begin{align*} 
	h(P\boxminus R) &= h(P)\cap h(R).
\end{align*}
\end{fact}

Recall that for sets $X$ and $Y$,
\begin{align*} 
	X- Y = X\cap Y^c.
\end{align*} Analogously, we have:

\begin{proposition} For any programs $P$ and $R$,
\begin{align} 
	\label{eq:P_boxminus_R=(P sqcap R)^copyright} P\boxminus R = P\sqcap R^\copyright.
\end{align}
\end{proposition}
\begin{proof} Using De Morgan's law, we compute
\begin{align*} 
	P\sqcap R^\copyright 
		&= \left\{a \leftarrow (B\cap C) \;\middle|\; a \leftarrow B\in P,\; a \leftarrow C\in R^\copyright \right\}\\
		&= \left\{a \leftarrow (B\cap C) \;\middle|\; a \leftarrow B\in P,\; a \leftarrow (A- C)\in R \right\}\\
		&= \left\{a \leftarrow (B\cap (A- C)) \;\middle|\; a \leftarrow B\in P,\; a \leftarrow C\in R \right\}\\
		&= \{a \leftarrow (B- C) \mid a \leftarrow B\in P,\; a \leftarrow C\in R\}\\
		&= P\boxminus R.
\end{align*}
\end{proof}

\begin{corollary}\label{c:(P_sqcup_R)_boxminus_Q} For any programs $P,Q,R$, we have
\begin{align}
	\label{eq:(P_sqcup_R)_boxminus_Q} (P\sqcup R)\boxminus Q \subseteq (P\boxminus Q)\sqcup (R\boxminus Q),
\end{align} with equality in case $Q$ is minimalist.
\end{corollary}
\begin{proof} Notice first that $Q$ is minimalist iff $Q^\copyright$ is minimalist, since body-complement preserves rule heads and neither adds nor removes rules. We can therefore compute
\begin{align*} 
	(P\sqcup R)\boxminus Q \stackrel{\ref{eq:P_boxminus_R=(P sqcap R)^copyright}}= (P\sqcup R)\sqcap Q^\copyright \stackrel{\ref{eq:(P_sqcup_R)_sqcap_Q}}\subseteq (P\sqcap Q^\copyright)\sqcup (R\sqcap Q^\copyright) \stackrel{\ref{eq:P_boxminus_R=(P sqcap R)^copyright}}= (P\boxminus Q)\sqcup (R\boxminus Q),
\end{align*} where the inclusion is an identity by \prettyref{eq:(P_sqcup_R)_sqcap_M} in case $Q$ --- and hence $Q^\copyright$ --- is minimalist.
\end{proof}

\begin{observation} The inclusion in \prettyref{eq:(P_sqcup_R)_boxminus_Q} can in general not be turned into an identity, not even modulo subsumption equivalence. For example, for $P := \{a \leftarrow b\}$ and $R := \{a \leftarrow c\}$ and the non-minimalist program $Q := \{a \leftarrow b,\; a \leftarrow c\}$, we have
\begin{align*} 
	(P\sqcup R)\boxminus Q = \left\{
	\begin{array}{l}
		a \leftarrow b\\
		a \leftarrow c
	\end{array}
	\right\} \subsetneq \left\{
	\begin{array}{l}
		\underline{a}\\
		a \leftarrow b\\
		a \leftarrow c\\
		\underline{a \leftarrow b,c}
	\end{array}
	\right\} = (P\boxminus Q)\sqcup (R\boxminus Q),
\end{align*} and since the program on the right contains the fact $a$ whereas the program on the left contains no fact at all, \prettyref{eq:P0=f(P)} yields
\begin{align*} 
	[(P\sqcup R)\boxminus Q]\emptyset = \emptyset \neq \{a\} = [(P\boxminus Q)\sqcup (R\boxminus Q)]\emptyset.
\end{align*} Of course, \prettyref{eq:(P_sqcup_R)_boxminus_Q} does imply
\begin{align*} 
	(P\sqcup R)\boxminus Q \lesssim (P\boxminus Q)\sqcup (R\boxminus Q).
\end{align*}
\end{observation}

Interestingly enough, we can compute the body-complement of $P$ via body-subtraction according to the formula
\begin{align*} 
	P^\copyright = A^\oplus\boxminus P.
\end{align*}

We shall now investigate the interaction between reducts and body-subtraction. We can express the head-reduct in terms of body-subtraction as (compare with \prettyref{eq:I_sqcup_P})
\begin{align} 
	\label{eq:P_boxminus_I} P\boxminus I &= {^I}P,
\end{align} for every program $P$ and interpretation $I$. Moreover, we can simplify the computation of body-subtraction using left reducts as
\begin{align*} 
	P\boxminus R = {^{h(R)}}P\boxminus {^{h(P)}}R.
\end{align*}

\begin{proposition} The head-reduct is an endomorphism ${^I}. : (\mathbb P_A,\boxminus)\to (\mathbb P_A,\boxminus)$ for every interpretation $I$, that is, for any programs $P$ and $R$,
\begin{align*} 
	{^I}(P\boxminus R) &= {^I}P\boxminus {^I}R.
\end{align*}
\end{proposition}
\begin{proof}
\begin{align*} 
	{^I}(P\boxminus R) 
		\stackrel{\ref{eq:I_sqcup_P}}= I\sqcup (P\boxminus R) 
		\stackrel{\ref{eq:P_boxminus_R=(P sqcap R)^copyright}}= I\sqcup (P\sqcap R^\copyright) 
		\stackrel{\ref{eq:(P_sqcap_R)_sqcup_M}}= (I\sqcup P)\sqcap (I\sqcup R^\copyright) 
		\stackrel{\ref{eq:I_sqcup_P}}= {^I}P\sqcap {^I}(R^\copyright) 
		\stackrel{\ref{eq:^IP^copyright=^IP^copyright}}= {^I}P\sqcap ({^I}R)^\copyright 
		\stackrel{\ref{eq:P_boxminus_R=(P sqcap R)^copyright}}= {^I}P\boxminus {^I}R.
\end{align*}
\end{proof}

We now turn our attention to the interaction between composition and body-subtraction.

\begin{observation} There is a minimalist Krom program $K$ and an interpretation $I$ such that
\begin{align*} 
	(K\boxminus K)I \not\approx KI\boxminus KI.
\end{align*} For example, we have
\begin{align*} 
	(\{a \leftarrow b\}\boxminus \{a \leftarrow b\}) \emptyset = \{a\} \emptyset = \{a\}
\end{align*} whereas
\begin{align*} 
	\{a \leftarrow b\} \emptyset\boxminus \{a \leftarrow b\} \emptyset = \emptyset\boxminus \emptyset = \emptyset.
\end{align*} Hence, in general we have
\begin{align*} 
	(P\boxminus R)Q \not\approx PQ\boxminus RQ.
\end{align*}
\end{observation}

\begin{observation} There are Krom programs $K$ and $L$ and an interpretation $I$ such that
\begin{align*} 
	K(L\boxminus I)\not\approx KL\boxminus KI.
\end{align*} For example, for
\begin{align*} 
	K := \left\{
	\begin{array}{l}
		a \leftarrow b\\
		a \leftarrow c
	\end{array}
	\right\} \quad\text{and}\quad L := \{b \leftarrow b\} \quad\text{and}\quad I := \{c\},
\end{align*} we have
\begin{align*} 
	\left\{
	\begin{array}{l}
		a \leftarrow b\\
		a \leftarrow c
	\end{array}
	\right\}(\{b \leftarrow b\}\boxminus \{c\}) = \emptyset
\end{align*} whereas
\begin{align*} 
	\left\{
	\begin{array}{l}
		a \leftarrow b\\
		a \leftarrow c
	\end{array}
	\right\}\{b \leftarrow b\}\boxminus \left\{
	\begin{array}{l}
		a \leftarrow b\\
		a \leftarrow c
	\end{array}
	\right\}\{c\} = \{a \leftarrow b\}\boxminus \{a\} = \{a \leftarrow b\}.
\end{align*}

Hence, in general
\begin{align*} 
	Q(P\boxminus R)\not\approx QP\boxminus QR.
\end{align*}
\end{observation}

\begin{fact} For any programs $P$ and $R$ with $h(P)\subseteq h(R)$,
\begin{align} 
	\label{eq:P_boxminus_R^omega} P^\omega\subseteq (P\boxminus R)^\omega.
\end{align}
\end{fact}

\section{Body-symmetric-difference}\label{sec:Body-symmetric-difference}

Recall that the set-theoretic symmetric difference of two programs is given by
\begin{align*} 
	P\Triangle R = (P- R)\cup (R- P).
\end{align*} We shall introduce analogous set-like operations on programs in two different ways corresponding to the two sides of the above equation: 

\begin{definition} Define the \emph{body-symmetric-difference} of $P$ and $R$ by
\begin{align*} 
	P\boxtriangle R &:= \{a \leftarrow (B\Triangle C)\mid a \leftarrow B\in P,\; a \leftarrow C\in R\},
\end{align*} and define the \emph{black-body-symmetric-difference} of $P$ and $R$ by
\begin{align*} 
	P \boxblacktriangle R := (P\boxminus R)\sqcup (R\boxminus P).
\end{align*}
\end{definition}

The two notions of ``symmetric difference'' coincide on minimalist programs containing at most one rule for each rule head, that is, for any minimalist programs $M$ and $N$ we clearly have
\begin{align*} 
	M\boxtriangle N = M\boxblacktriangle N.
\end{align*} 

However, for non-minimalist programs, the two notions may differ as is witnessed by the following example: 

\begin{example} For
\begin{align*} 
	P := \left\{
	\begin{array}{l}
		a \leftarrow b\\
		a \leftarrow c
	\end{array}
	\right\} \quad\text{and}\quad R := \left\{
	\begin{array}{l}
		a \leftarrow b\\
		a \leftarrow d
	\end{array}
	\right\}
\end{align*} we have
\begin{align*} 
	P\boxtriangle R = \left\{
	\begin{array}{l}
		a\\
		a \leftarrow b,c\\
		a \leftarrow b,d\\
		a \leftarrow c,d
	\end{array}
	\right\}
\end{align*} whereas
\begin{align*} 
	P\boxblacktriangle R = \left\{
	\begin{array}{l}
		a\\
		a \leftarrow b\\
		a \leftarrow c
	\end{array}
	\right\} \sqcup \left\{
	\begin{array}{l}
		a\\
		a \leftarrow b\\
		a \leftarrow d
	\end{array}
	\right\} = \left\{
	\begin{array}{l}
		a\\
		\underline{a \leftarrow b}\\
		\underline{a \leftarrow c}\\
		\underline{a \leftarrow d}\\
		a \leftarrow b,c\\
		a \leftarrow b,d\\
		a \leftarrow c,d
	\end{array}
	\right\}.
\end{align*}
\end{example}

\begin{proposition} For any programs $P$ and $R$,
\begin{align*} 
	P \boxtriangle R \subseteq P\boxblacktriangle R.
\end{align*}
\end{proposition}
\begin{proof} 
\begin{align*} 
	P \boxtriangle R &= \{a \leftarrow (B\Triangle C) \mid a \leftarrow B\in P,\; a \leftarrow C\in R\}\\
		&= \{a \leftarrow ((B-C)\cup (C-B)) \mid a \leftarrow B\in P,\; a \leftarrow C\in R\}\\
		&\subseteq \{a \leftarrow ((B-C)\cup (D-E)) \mid a \leftarrow B\in P,\; a \leftarrow C\in R,\; a \leftarrow D\in R,\; a \leftarrow E\in P\}\\
		&=(P\boxminus R)\sqcup (R\boxminus P)\\
		&= P\boxblacktriangle R.
\end{align*}
\end{proof}

\section{Body-power-set}\label{sec:Body-power-set}

This section introduces and studies a set-like operation on programs inspired by the analogous power set operation on sets:

\begin{definition} Define the \emph{body-power-set} of $P$ by
\begin{align*} 
	P^{ \mathfrak p} := \{a \leftarrow B \mid a \leftarrow C\in P,\; B\subseteq C\}.
\end{align*}
\end{definition}
 
For example, we have
\begin{align*} 
	\emptyset^{ \mathfrak p} = \emptyset,\quad \{a\}^{ \mathfrak p} = \{a\},\quad \left\{
	\begin{array}{l}
		a\\
		b \leftarrow a\\
		c \leftarrow a,b
	\end{array}
	\right\}^{ \mathfrak p} = \left\{
	\begin{array}{l}
		a\\
		b\\
		c\\
		b \leftarrow a\\
		c \leftarrow a\\
		c \leftarrow b\\
		c \leftarrow a,b
	\end{array}
	\right\}.
\end{align*}

Interestingly, we have
\begin{align}\label{eq:P^p=P-F}
	P^{ \mathfrak p} = P\boxminus F_A,
\end{align} where $F_A$ is the full program defined in \prettyref{eq:F_A}.

\begin{fact} The body-power-set operator is an endomorphism on $(\mathbb P_A,\cup,\emptyset)$\footnote{Recall that $F_A$ is the full program over $A$ defined in \prettyref{eq:F_A}.} thus satisfying
\begin{align*} 
	(P\cup R)^{ \mathfrak p} &= P^{ \mathfrak p}\cup R^{ \mathfrak p}\\
	\emptyset^{ \mathfrak p} &= \emptyset.
\end{align*} Moreover, we have
\begin{align*} 
	(P\cap R)^{ \mathfrak p} &\subseteq P^{ \mathfrak p}\cap R^{ \mathfrak p}\\
	F_A^{ \mathfrak p} &= F_A.
\end{align*}
\end{fact}

Computing the least model of $P^{ \mathfrak p}$ is easy as it contains a fact for each rule head of $P$ and therefore
\begin{align*} 
	(P^{ \mathfrak p})^\omega = h(P).
\end{align*}

The body-power-set operator is idempotent and the identity on interpretations, that is,
\begin{align} 
	(P^{ \mathfrak p})^{ \mathfrak p} &= P^{ \mathfrak p}\\
	\label{eq:I^p=I} I^{ \mathfrak p} &= I.
\end{align} 

Moreover, it can be written in terms of composition as
\begin{align} 
	\label{eq:P^p} P^{ \mathfrak p} \stackrel{\ref{eq:I^ast}}= PA^\ast.
\end{align} 

In particular, for any Krom program $K$, we can derive that identity as
\begin{align*} 
	K^{ \mathfrak p} = K\cup h(K) \stackrel{\ref{eq:h(P)=PA}}= K\cup KA \stackrel{\ref{eq:K(P_cup_R)}}= K(1\cup A) \stackrel{\ref{eq:I^ast}}= KA^\ast.
\end{align*} As a consequence, we have 
\begin{align*} 
	(KL)^{ \mathfrak p} \stackrel{\ref{eq:P^p}}= (KL)A^\ast \stackrel{\ref{eq:KPR}}= K(LA^\ast) \stackrel{\ref{eq:P^p}}= KL^{ \mathfrak p},
\end{align*} where the second identity follows by the associativity of composition with respect to Krom programs (recall that composition is in general not associative). 

Similarly,
\begin{align*} 
	(PR)^{ \mathfrak p} \stackrel{\ref{eq:P^p}}= (PR)A^\ast\approx P(RA^\ast) \stackrel{\ref{eq:P^p}}= PR^{ \mathfrak p},
\end{align*} where we have to invoke \prettyref{p:seminearring} instead of \prettyref{eq:KPR} as for Krom programs above since composition is in general associative only modulo subsumption equivalence.

An interesting identity is
\begin{align*} 
	1_A^{ \mathfrak p} = A^\ast.
\end{align*}

For the program $A^\oplus$, neutral to body-intersection \prettyref{eq:P_sqcap_A^oplus}, we have
\begin{align*} 
	(A^\oplus)^{ \mathfrak p} = F_A.
\end{align*} Hence, as a consequence of \prettyref{eq:P^p=P-F},
\begin{align*} 
	P^{ \mathfrak p } = P\boxminus (A^\oplus)^{ \mathfrak p}.
\end{align*}

We now analyze the body-power-set operator in combination with body-union. 
The next result shows that the body-power-set operator is a homomorphism for body-union:

\begin{proposition}\label{p:s_morphism} The body-power-set operator is an endomorphism on the monoid $(\mathbb P_A,\sqcup,A)$, that is,
\begin{align*} 
	(P\sqcup R)^{ \mathfrak p} &= P^{ \mathfrak p}\sqcup R^{ \mathfrak p}\\
	A^{ \mathfrak p} & \stackrel{\ref{eq:I^p=I}}= A.
\end{align*}
\end{proposition}
\begin{proof} Since
\begin{align*} 
	A^\ast\sqcup A^\ast \stackrel{\ref{eq:I^ast}}= (1_A\cup A)\sqcup (1_A\cup A) = 1_A\cup A \stackrel{\ref{eq:I^ast}}= A^\ast,
\end{align*} we may apply \prettyref{eq:(P_sqcup_R)Q=PQ_sqcup_RQ} and compute
\begin{align*} 
	(P\sqcup R)^{ \mathfrak p} \stackrel{\ref{eq:P^p}}= (P\sqcup R)A^\ast \stackrel{\ref{eq:(P_sqcup_R)Q=PQ_sqcup_RQ}}= PA^\ast\sqcup RA^\ast \stackrel{\ref{eq:P^p}}= P^{ \mathfrak p}\sqcup R^{ \mathfrak p}.
\end{align*}
\end{proof}

\section{Homomorphisms}

This section is concerned with homomorphisms, which are structure-preserving mappings between algebras.

First, we state that interpretations under the usual set operations can be embedded into programs under the set-like operations as follows:

\begin{proposition} The mapping ${}^\square:\mathbb I_A\to \mathbb P_A$ given by
\begin{align*} 
	I^\square := \{a \leftarrow I \mid a\in A\}
\end{align*} is a monomorphism $(\mathbb I_A,\cup,\cap,{}^c,\emptyset,A)\to (\mathbb M_A,\sqcup,\sqcap,{}^\copyright,A,A^\oplus)$, that is, ${}^\square$ maps every interpretation $I\in \mathbb I_A$ to a minimalist program $I^\square\in \mathbb M_A$ such that
\begin{align} 
	(I\cup J)^\square &= I^\square\sqcup J^\square\\
	(I\cap J)^\square &= I^\square\sqcap J^\square\\
	(I^c)^\square &= (I^\square)^\copyright\\
	\label{eq:emptyset^square=A} \emptyset^\square &= A\\
	\label{eq:A^square=A^oplus} A^\square &= A^\oplus.
\end{align}
\end{proposition}
\begin{proof} Follows immediately from the definitions.
\end{proof}

\begin{remark} Equations \prettyref{eq:emptyset^square=A} and \prettyref{eq:A^square=A^oplus} mean that $\emptyset^\square$ and $A^\square$ are the neutral elements $A$ and $A^\oplus$ with respect to $\sqcup$ and $\sqcap$, respectively.
\end{remark}

The following theorem gives us a way to construct homomorphisms bottom-up from single-element interpretations representing single atoms:

\begin{theorem}\label{t:bullet} Let $A$ and $A'$ be alphabets and let ${}^\bullet:\mathbb I_A\to \mathbb I_{A'}$ be a mapping on interpretations satisfying
\begin{align} 
	\label{eq:I^bullet} I^\bullet &= \bigcup_{a\in I}\{a\}^\bullet,\\
	\label{eq:a^bullet_neq_emptyset} \{a\}^\bullet &\neq \emptyset,\quad\text{for every $a\in A$},\\
	\label{eq:a^bullet_cap_b^bullet} \{a\}^\bullet\cap \{b\}^\bullet &= \emptyset,\quad\text{for all $a,b\in A$ with $a\neq b$}.
\end{align} We extend ${}^\bullet$ to a mapping $\mathbb P_A\to \mathbb P_{A'}$ on programs as
\begin{align*}
	P^\bullet := \left\{a' \leftarrow B^\bullet \;\middle|\; a \leftarrow B\in P,\, a'\in \{a\}^\bullet\right\}.
\end{align*} Then
\begin{align*} 
	{}^\bullet : (\mathbb P_A,\cup,\sqcup, \emptyset)\to (\mathbb P_{A'},\cup,\sqcup, \emptyset)
\end{align*} is a homomorphism satisfying
\begin{align*} 
	P\in \mathbb I_A \quad\Leftrightarrow\quad P^\bullet\in \mathbb I_{A'},
\end{align*} that is, a program $P\in \mathbb P_A$ is an interpretation iff $P^\bullet\in \mathbb P_{A'}$ is an interpretation.
\end{theorem}
\begin{proof} Instantiating \prettyref{eq:I^bullet} with $I=\emptyset$ yields $\emptyset^\bullet = \emptyset$, and ${}^\bullet$ is a homomorphism with respect to $\cup$ since
\begin{align*} 
	(P\cup R)^\bullet 
		&= \left\{a' \leftarrow B^\bullet \;\middle|\; a \leftarrow B\in P\cup R,\; a'\in \{a\}^\bullet\right\}\\
		&= \left\{a' \leftarrow B^\bullet \;\middle|\; a \leftarrow B\in P,\; a'\in \{a\}^\bullet\right\}\cup \left\{a' \leftarrow B^\bullet \;\middle|\; a \leftarrow B\in R,\; a'\in \{a\}^\bullet\right\}\\
		&= P^\bullet\cup R^\bullet.
\end{align*} 

To show that ${}^\bullet$ is a homomorphism with respect to $\sqcup$, we compute
\begin{align*} 
	(P\sqcup R)^\bullet 
		&= \{a \leftarrow (B\cup C) \mid a \leftarrow B\in P,\; a \leftarrow C\in R\}^\bullet\\
		&= \left\{a' \leftarrow (B\cup C)^\bullet \;\middle|\; a \leftarrow B\in P,\; a \leftarrow C\in R,\; a'\in \{a\}^\bullet\right\}\\
		& \stackrel{\ref{eq:I^bullet}}= \left\{a' \leftarrow (B^\bullet\cup C^\bullet) \;\middle|\; a \leftarrow B\in P,\; a \leftarrow C\in R,\; a'\in \{a\}^\bullet\right\}\\
		& \stackrel{\ref{eq:a^bullet_cap_b^bullet}}= P^\bullet\sqcup R^\bullet,
\end{align*} where the last identity holds since $P^\bullet\sqcup R^\bullet$ consists of all rules $a' \leftarrow (D\cup E)$ with $a' \leftarrow D\in P^\bullet$ and $a' \leftarrow E\in R^\bullet$: the inclusion from left to right is immediate from the definition of ${}^\bullet$, and for the inclusion from right to left notice that by \prettyref{eq:a^bullet_cap_b^bullet} the head $a'$ determines the unique atom $a$ satisfying $a'\in \{a\}^\bullet$, which means that $D=B^\bullet$ and $E=C^\bullet$ for rules $a \leftarrow B\in P$ and $a \leftarrow C\in R$ with one and the same head $a$.

It remains to prove the stated equivalence. Every interpretation $I\in \mathbb I_A$ is a program consisting of facts only, which by \prettyref{eq:I^bullet} gives
\begin{align*} 
	\{a \leftarrow \emptyset\mid a\in I\}^\bullet = \left\{a' \leftarrow \emptyset^\bullet \;\middle|\; a\in I,\; a'\in \{a\}^\bullet\right\} = \left\{a' \leftarrow \emptyset \;\middle|\; a'\in \bigcup_{a\in I}\{a\}^\bullet\right\} \stackrel{\ref{eq:I^bullet}}= I^\bullet,
\end{align*} which shows that the extension of ${}^\bullet$ to programs is consistent with the given mapping on interpretations and, in particular, that $P^\bullet$ is an interpretation whenever $P$ is one. For the other direction, assume that $P$ is not an interpretation and pick a proper rule $a \leftarrow B\in p(P)$ together with an atom $b\in B$. By \prettyref{eq:a^bullet_neq_emptyset}, we can pick atoms $a'\in \{a\}^\bullet$ and $b'\in \{b\}^\bullet$, and \prettyref{eq:I^bullet} yields $b'\in B^\bullet$. Hence, $a' \leftarrow B^\bullet$ is a proper rule in $P^\bullet$ and therefore $P^\bullet$ is not an interpretation.
\end{proof}

\begin{example} Let $A := \{a,b\}$ and $A' := \{a',b',c'\}$. Define the mapping ${}^\bullet$ on interpretations in $\mathbb I_A$ by
\begin{align*} 
	\emptyset^\bullet &:= \emptyset\\
	\{a\}^\bullet &:= \{a',c'\}\\
	\{b\}^\bullet &:= \{b'\}.
\end{align*} This forces us to put
\begin{align*} 
	\{a,b\}^\bullet := \{a\}^\bullet\cup \{b\}^\bullet = \{a',b',c'\}.
\end{align*} Informally speaking, ${}^\bullet$ ``blows up'' the atom $a$ to $a',c'$ and leaves the atom $b$ ``unchanged''. By \prettyref{t:bullet}, extending ${}^\bullet$ to programs as described above we automatically get a homomorphism satisfying, for example,
\begin{align*} 
	\left(\left\{
	\begin{array}{l}
		a\\
		b \leftarrow a
	\end{array}
	\right\}\sqcup \left\{
	\begin{array}{l}
		a\\
		b \leftarrow b
	\end{array}
	\right\}\right)^\bullet &= \left\{
	\begin{array}{l}
		a\\
		b \leftarrow a,b
	\end{array}
	\right\}^\bullet\\ 
	&= \left\{
	\begin{array}{l}
		a'\\
		c'\\
		b' \leftarrow a',b',c'
	\end{array}
	\right\}\\ 
	&= \left\{
	\begin{array}{l}
		a'\\
		c'\\
		b' \leftarrow a',c'
	\end{array}
	\right\}\sqcup \left\{
	\begin{array}{l}
		a'\\
		c'\\
		b' \leftarrow b'
	\end{array}
	\right\}\\ 
	&= \left\{
	\begin{array}{l}
		a\\
		b \leftarrow a
	\end{array}
	\right\}^\bullet\sqcup \left\{
	\begin{array}{l}
		a\\
		b \leftarrow b
	\end{array}
	\right\}^\bullet.
\end{align*}
\end{example}

\section{Decomposition}\label{sec:Decomposition}

This section is concerned with syntactic and semantic decompositions of programs. 

Our first observation is as follows:

\begin{fact} We can syntactically decompose every program $P$ into single-rule Krom programs as
\begin{align} 
	\label{eq:P=bigcup_bigsqcup} P = f(P)\cup \bigcup_{a \leftarrow B\in p(P)}\bigsqcup_{b\in B}\{a \leftarrow b\}.
\end{align}
\end{fact}

For example,
\begin{align*} 
	\left\{
	\begin{array}{l}
		a\\
		b \leftarrow a,b\\
		c \leftarrow a,b,c
	\end{array}
	\right\} = \{a\}\cup (\{b \leftarrow a\}\sqcup \{b \leftarrow b\})\cup (\{c \leftarrow a\}\sqcup \{c \leftarrow b\}\sqcup \{c \leftarrow c\}).
\end{align*} 

In what follows, we are interested in decomposing programs with respect to body-union. 

We shall now show that composition can be represented via $\sqcup$, $\cup$, $f$, and $p$ which is useful given that $\sqcup$ and $\cup$ satisfy nice algebraic laws (see Theorems \ref{t:cup_sqcup} and \ref{t:DT_cup_ss}):

\begin{theorem}\label{t:bigcup_bigsqcup_bigcup_bigsqcup} For any programs $P$ and $R$,
\begin{align} 
	\label{eq:PR=bigcup_bigsqcup_bigcup_bigsqcup} PR = f(P)\cup \bigcup_{a \leftarrow B\in p(P)}\bigsqcup_{b\in B}\bigcup_{b \leftarrow C\in R}\bigsqcup_{c\in C}\{a \leftarrow c\},
\end{align} where we use the convention $\bigsqcup_{c\in C}\{a \leftarrow c\} = \{a \leftarrow C\}$ for every atom $a$ and every set of atoms $C$, so that the empty body-union of rules with head $a$ is the fact $a$.
\end{theorem}
\begin{proof} 
\begin{align*} 
	PR &= f(P)\cup \bigcup_{a \leftarrow B\in p(P)}\{a \leftarrow B\}R\\
		&= f(P)\cup \bigcup_{a \leftarrow B\in p(P)}\bigsqcup_{b\in B}\{a \leftarrow b\}R\\
		& \stackrel{\ref{eq:K(P_cup_R)}}= f(P)\cup  \bigcup_{a \leftarrow B\in p(P)}\bigsqcup_{b\in B}\bigcup_{c \leftarrow C\in R}\{a \leftarrow b\}\{c \leftarrow C\}\\
		&= f(P)\cup \bigcup_{a \leftarrow B\in p(P)}\bigsqcup_{b\in B}\bigcup_{b \leftarrow C\in R}\{a \leftarrow C\}\\
		&= f(P)\cup \bigcup_{a \leftarrow B\in p(P)}\bigsqcup_{b\in B}\bigcup_{b \leftarrow C\in R}\bigsqcup_{c\in C}\{a \leftarrow c\},
\end{align*} where the first identity uses $f(P)R=f(P)$ together with \prettyref{eq:(P_cup_R)Q}. For the second identity, we first write $\{a\leftarrow B\}=\bigsqcup_{b\in B}\{a\leftarrow b\}$ and then apply \prettyref{eq:(P_sqcup_R)Q=PQ_sqcup_RQ} iteratively. This is legitimate since the bodies $\{b\}$, $b\in B$, are pairwise disjoint and therefore satisfy case (i).
\end{proof}

Recall from \prettyref{eq:K(P_cup_R)} that Krom programs have the nice property that composition distributes over union from both sides which is not the case in general (see \prettyref{eq:P(Q_cup_R)}). Decomposing programs into Krom programs is thus desirable. We shall now prove the main results of the paper showing that every program can be $\sqcup$-decomposed into Krom programs (Theorems \ref{t:MDT} and \ref{t:DT}), but before we can give the formal proofs we first introduce some auxiliary notions and then illustrate the main idea with an example:

\begin{definition} Define the \emph{body-or} of $P$ as
\begin{align*} 
	P^\lor := f(P)\cup \{a \leftarrow b \mid a \leftarrow B\in p(P),\; b\in B\}.
\end{align*} 
\end{definition}

Notice that the body-or of a program is always a Krom program consisting only of rules with at most one body atom.
 For example, we have
\begin{align*} 
	\left\{
	\begin{array}{l}
		a\\
		a \leftarrow b,c
	\end{array}
	\right\}^\lor = \left\{
	\begin{array}{l}
		a\\
		a \leftarrow b\\
		a \leftarrow c
	\end{array}
	\right\}.
\end{align*}



\begin{example}\label{e:M} Now consider the minimalist program
\begin{align*} 
	M := \left\{
	\begin{array}{l}
		a\\
		b\\
		c \leftarrow a,b\\
		d \leftarrow a,b,d
	\end{array}
	\right\}.
\end{align*} We show that we can decompose $M$ with width $w(M)=3$ into three Krom programs $K_1,K_2,K_3$, where
\begin{align*} 
	K_1,K_2,K_3\subseteq M^\lor = \left\{
	\begin{array}{l}
		a\\
		b\\
		c \leftarrow a\\
		c \leftarrow b\\
		d \leftarrow a\\
		d \leftarrow b\\
		d \leftarrow d
	\end{array}
	\right\},
\end{align*} such that
\begin{align} 
	\label{eq:M=K1K2K3} M = K_1\sqcup K_2\sqcup K_3
\end{align} and
\begin{align*} 
	M^\omega = K_1^\omega\sqcup K_2^\omega\sqcup K_3^\omega.
\end{align*} For this, define
\begin{align*} 
	K_1 := \left\{
	\begin{array}{l}
		a\\
		b\\
		c \leftarrow a\\
		d \leftarrow a 
	\end{array}
	\right\},\quad K_2 := \left\{
	\begin{array}{l}
		a\\
		b\\
		c \leftarrow b\\
		d \leftarrow b
	\end{array}
	\right\},\quad K_3 := \left\{
	\begin{array}{l}
		a\\
		b\\
		c \leftarrow b\\
		d \leftarrow d
	\end{array}
	\right\}.
\end{align*} We have \prettyref{eq:M=K1K2K3} and
\begin{align*} 
	M^\omega = \{a,b,c\} = \{a,b,c,d\}\cap \{a,b,c,d\}\cap \{a,b,c\} = K_1^\omega\cap K_2^\omega\cap K_3^\omega \stackrel{\ref{eq:I_sqcup_J}}=K_1^\omega\sqcup K_2^\omega\sqcup K_3^\omega.
\end{align*}
\end{example}

We are now ready to prove the following decomposition theorem for minimalist programs containing at most one rule for each rule head by generalizing the reasoning in \prettyref{e:M}:

\begin{theorem}[Minimalist Decomposition Theorem]\label{t:MDT} Let $M$ be a minimalist program and put
\begin{align*} 
	m := \max\{1,w(M)\}.
\end{align*} Then there are Krom programs
\begin{align*} 
	K_1,\ldots,K_m\subseteq M^\lor
\end{align*} such that
\begin{align}
    \label{eq:M=K1_sqcups_Km} M = K_1\sqcup\ldots\sqcup K_m
\end{align} and
\begin{align}
    \label{eq:Momega=intersection} M^\omega = K_1^\omega\sqcup\ldots\sqcup K_m^\omega \stackrel{\ref{eq:I_sqcup_J}}= K_1^\omega\cap\cdots\cap K_m^\omega.
\end{align}
\end{theorem}
\begin{proof}
Since $M$ is minimalist, for every $a\in h(M)$ there is a unique rule
\begin{align*} 
	r_a = a\leftarrow B_a
\end{align*} in $M$.

We first choose, for every non-empty body $B_a$, a distinguished atom $b_a^\ast\in B_a$. If $a\notin M^\omega$, choose
\begin{align*} 
	b_a^\ast\in B_a- M^\omega.
\end{align*} Such an atom exists. Indeed, if $B_a\subseteq M^\omega$, then
\begin{align*} 
	a\in MM^\omega = M^\omega,
\end{align*} contrary to $a\notin M^\omega$. If $a\in M^\omega$, choose $b_a^\ast$ arbitrarily from $B_a$.

Since $|B_a|\leq m$, for every non-empty $B_a$ we can choose
atoms
\begin{align*} 
	b_{a,1},\ldots,b_{a,m}\in B_a
\end{align*} such that
\begin{align*} 
	b_{a,1}=b_a^\ast \qquad\text{and}\qquad B_a=\{b_{a,1},\ldots,b_{a,m}\}.
\end{align*} If $|B_a|<m$, the remaining positions are filled by repeating arbitrary atoms of $B_a$.

For $1\leq i\leq m$, define
\begin{align*} 
	K_i := f(M)\cup \{a\leftarrow b_{a,i} \mid a\leftarrow B_a\in p(M)\}.
\end{align*} Every $K_i$ is a Krom program. Moreover, every rule $a\leftarrow b_{a,i}$ belongs to $M^\lor$, and hence
\begin{align*} 
	K_i\subseteq M^\lor.
\end{align*}

For every $a\in h(M)$, the unique rules with head $a$ in the programs $K_1,\ldots,K_m$ have bodies whose union is $B_a$. Consequently,
\begin{align*} 
	K_1\sqcup\ldots\sqcup K_m = \{a\leftarrow B_a\mid a\in h(M)\} = M.
\end{align*} This proves \eqref{eq:M=K1_sqcups_Km}.

We next compare the least models. For every interpretation $I$ and every $i$, we have
\begin{align*} 
	MI\subseteq K_i I.
\end{align*} Indeed, whenever $B_a\subseteq I$, we also have $b_{a,i}\in I$. Facts are present in both $M$ and $K_i$. By induction on the iterations of the immediate consequence operators, it follows that
\begin{align*} 
	M^\omega\subseteq K_i^\omega \qquad (1\leq i\leq m).
\end{align*} It remains to show that the first component introduces no new atoms. We claim that $M^\omega$ is a model of $K_1$. Consider a rule in $K_1$. If it is a fact, then its head belongs to $f(M)\subseteq M^\omega$. Otherwise, it has the form
\begin{align*} 
	a\leftarrow b_{a,1}.
\end{align*} If $a\in M^\omega$, this rule is satisfied by $M^\omega$. If $a\notin M^\omega$, then the construction of $b_{a,1}$ gives
\begin{align*} 
	b_{a,1}=b_a^\ast\notin M^\omega,
\end{align*} so the body of the rule is false in $M^\omega$. Thus every rule
of $K_1$ is satisfied by $M^\omega$, and therefore
\begin{align*} 
	K_1^\omega\subseteq M^\omega.
\end{align*} Together with the opposite inclusion proved above, this yields
\begin{align*} 
	K_1^\omega=M^\omega.
\end{align*}

Finally,
\begin{align*} 
	M^\omega=K_1^\omega\subseteq K_i^\omega \qquad (1\leq i\leq m),
\end{align*} and hence
\begin{align*} 
	K_1^\omega\cap\ldots\cap K_m^\omega=M^\omega.
\end{align*} Since body-union coincides with intersection on interpretations by \prettyref{eq:I_sqcup_J},
we obtain
\begin{align*} 
	K_1^\omega\sqcup\ldots\sqcup K_m^\omega = K_1^\omega\cap\ldots\cap K_m^\omega = M^\omega.
\end{align*} This proves \eqref{eq:Momega=intersection}.
\end{proof}

\begin{remark} The decomposition in \prettyref{eq:M=K1_sqcups_Km} is not unique. The condition in \prettyref{eq:M=K1_sqcups_Km} is essential since without it, we can simply choose $M^\omega$ --- which is an interpretation and thus trivially Krom --- as our ``decomposition.''
\end{remark}

\begin{theorem}[Decomposition Theorem]\label{t:DT} Let $P$ be a program and put
\begin{align*} 
	m := \max\{1,w(P)\}.
\end{align*} Then there are Krom programs
\begin{align*} 
	K_1,\ldots,K_m\subseteq P^\lor
\end{align*} such that
\begin{align}
	\label{eq:P_subseteq_K1_sqcups_Km} P &\subseteq K_1\sqcup\cdots\sqcup K_m
\end{align} and
\begin{align}
	\label{eq:Pomega_subseteq_intersection} P^\omega &\subseteq K_1^\omega\sqcup\cdots\sqcup K_m^\omega \stackrel{\ref{eq:I_sqcup_J}}= K_1^\omega\cap\cdots\cap K_m^\omega.
\end{align}
\end{theorem}
\begin{proof} Since $|B|\leq m$ holds for every rule $a \leftarrow B\in p(P)$, we can choose for every such rule $r = a \leftarrow B$ atoms
\begin{align*} 
	b_{r,1},\ldots,b_{r,m}\in B \qquad\text{such that}\qquad B = \{b_{r,1},\ldots,b_{r,m}\},
\end{align*} where in case $|B|<m$ the remaining positions are filled by repeating arbitrary atoms of $B$. For $1\leq i\leq m$, define
\begin{align*} 
	K_i := f(P)\cup \{a \leftarrow b_{r,i} \mid r = a \leftarrow B\in p(P)\}.
\end{align*} Every $K_i$ is a Krom program and every one of its rules occurs in $P^\lor$, which gives $K_i\subseteq P^\lor$.

We first prove \prettyref{eq:P_subseteq_K1_sqcups_Km}. Let $a \leftarrow B\in P$. In case $B=\emptyset$, we have $a\in f(P)\subseteq K_i$ for every $i$ and hence $a\in K_1\sqcup\cdots\sqcup K_m$. In case $r = a \leftarrow B\in p(P)$, the rule $a \leftarrow b_{r,i}$ belongs to $K_i$ for every $i$, and the body-union of these $m$ rules is
\begin{align*} 
	a \leftarrow (\{b_{r,1}\}\cup\ldots\cup \{b_{r,m}\}) = a \leftarrow B,
\end{align*} which shows $a \leftarrow B\in K_1\sqcup\cdots\sqcup K_m$. Notice that in contrast to \prettyref{t:MDT} we cannot expect equality here, since the body-union also combines rules originating from \emph{different} rules of $P$ with the same head, thus producing rules not occurring in $P$ (see \prettyref{e:DT}).

We now prove \prettyref{eq:Pomega_subseteq_intersection}. Let $I$ be an arbitrary interpretation and let $a\in PI$, that is, there is a rule $a \leftarrow B\in P$ with $B\subseteq I$. In case $B=\emptyset$, we have $a\in f(P)\subseteq K_i$ and hence $a\in K_iI$; and in case $r = a \leftarrow B\in p(P)$, we have $b_{r,i}\in B\subseteq I$ and therefore $a\in K_iI$ as well. This shows
\begin{align*} 
	P\lesssim K_i \qquad (1\leq i\leq m),
\end{align*} which by \prettyref{eq:P_leqq_ss_R_} and \prettyref{eq:LM(P)=P^omega} yields $P^\omega\subseteq K_i^\omega$ for every $i$ and hence
\begin{align*} 
	P^\omega\subseteq K_1^\omega\cap\cdots\cap K_m^\omega \stackrel{\ref{eq:I_sqcup_J}}= K_1^\omega\sqcup\cdots\sqcup K_m^\omega.
\end{align*}
\end{proof}

The next example provides an intuition why we cannot turn the inclusion in \prettyref{t:DT} into an equality:

\begin{example}\label{e:DT} The non-minimalist program
\begin{align*} 
	P := \left\{
	\begin{array}{l}
		a \leftarrow b,c\\
		a \leftarrow d,e
	\end{array}
	\right\}
\end{align*} cannot be $\sqcup$-decomposed into Krom programs. Consider, for instance,
\begin{align*} 
	\left\{
	\begin{array}{l}
		a \leftarrow b\\
		a \leftarrow d
	\end{array}
	\right\}\sqcup \left\{
	\begin{array}{l}
		a \leftarrow c\\
		a \leftarrow e
	\end{array}
	\right\} = \left\{
	\begin{array}{l}
		a \leftarrow b,c\\
		a \leftarrow \underline{b,e}\\
		a \leftarrow \underline{d,c}\\
		a \leftarrow d,e
	\end{array}
	\right\},
\end{align*} where the body-union produces the two underlined rules not occurring in $P$. This is not an artifact of the particular choice. Suppose $P = K_1\sqcup\ldots\sqcup K_n$ for Krom programs $K_1,\ldots,K_n$. Since every rule in $P$ has head $a$, each $K_i$ contains at least one rule with head $a$, and by the definition of body-union there are bodies $B_1,\ldots,B_n$ and $C_1,\ldots,C_n$ of rules with head $a$ in $K_1,\ldots,K_n$ satisfying
\begin{align*} 
	B_1\cup\ldots\cup B_n = \{b,c\} \quad\text{and}\quad C_1\cup\ldots\cup C_n = \{d,e\}.
\end{align*} As each $B_i$ contains at most one atom, we have $B_i = \{b\}$ for some $i$, and replacing $C_i$ by $B_i$ in the second selection yields a rule in $K_1\sqcup\ldots\sqcup K_n = P$ whose body contains $b$ and is contained in $\{b,d,e\}$ --- contradicting the fact that the only rule in $P$ with $b$ occurring in its body is $a \leftarrow b,c$.
\end{example}

\section{Towards applications to automatic logic programming}\label{sec:App}

Automatic logic programming studies the construction of logic programs from specifications and background knowledge \citep{Deville94}. A key challenge in this setting is to provide principled mechanisms for combining and transforming programs while preserving their semantics. 

In this section, we sketch how the set-like operations introduced in this paper can contribute to this goal by enabling structured manipulations of programs and supporting decomposition-based approaches to program construction.

Beyond these examples, algebraic operations on logic programs are also relevant to program transformation and optimization \citep{Fuchs92,Partsch83}.

\subsection{Algebraic semantics}

A central requirement for algebraic reasoning about logic programs is the ability to decompose complex programs into simpler components in such a way that their semantics can be derived from the semantics of these components. 

Theorems \ref{t:MDT} and \ref{t:DT} provide two related structural results: minimalist programs admit exact representations in terms of Krom components, whereas arbitrary programs admit corresponding syntactic and semantic upper approximations. In the minimalist case, the least model is the intersection of the least models of the components.

Krom programs satisfy algebraic laws such as associativity \prettyref{eq:KPR} and left-distributivity \prettyref{eq:K(P_cup_R)}, which fail for arbitrary programs. The decomposition results therefore provide structural representations in terms of algebraically simpler components. Since the construction used in \prettyref{t:MDT} depends on the least model of the original program, however, the theorem does not by itself yield an algorithmic reduction of reasoning to the components.

\subsection{Schema-based transformation}

Program schemes provide a common mechanism for representing reusable program fragments that can be instantiated to obtain concrete programs \citep{Greibach75,Flener95,Fuchs92}. 

Within an algebraic setting, such schemes can be naturally understood as terms (or polynomials) over logic programs, where programs act as constants, algebraic operations as function symbols, and program variables as placeholders (see e.g. \citep[§II.1]{Burris00}). For example, the term $\{a \leftarrow b\}\sqcup X$ defines a schema that can be instantiated with a program $P$ to yield $\{a \leftarrow b\}\sqcup P$, thus inducing a transformation on programs.

Crucially, the expressive power of such schemes depends on the underlying algebra of programs. Introducing new algebraic operations expands the space of definable transformations. In particular, the set-like operations considered in this paper enable schemes that are not expressible in a natural way without them, such as those based on body-union.

\subsection{Inductive logic programming}\label{sec:ILP}

In inductive logic programming (ILP) \citep{Muggleton91,Shapiro81,Cropper22,Nienhuys-Cheng97}, (partial) specifications consist of sets of positive and negative examples in the form of interpretations $I_+$ and $I_-$, respectively, and the task is, given some background program $R$, to find a program $P$ such that
\begin{align*} 
	I_+\subseteq (P\cup R)^\omega \quad\text{and}\quad I_-\cap (P\cup R)^\omega = \emptyset.
\end{align*}

From an algebraic perspective, this task amounts to searching the space of programs under the least model operator and its interaction with program composition. A key difficulty is that this search space is typically large and poorly structured.

Algebraic operations on programs provide a way to impose structure on this space. In particular, the decomposition results established in Theorems \ref{t:MDT} and \ref{t:DT} allow one to reduce the search for a solution to reasoning about simpler components, since they make explicit how the behavior of a program emerges from the behavior of its parts.

\subsection{Learning from interpretation transitions}\label{sec:LfIT}

Closely related to ILP \prettyref{sec:ILP} is the following learning setting introduced in \cite{Inoue14}. 

Every logic program $P$ is a state transition system that maps, via its immediate consequence van Emden-Kowalski operator $T_P$, every state of the world represented as an interpretation $I$ to a next state $T_P(I)$ \citep{Inoue11,Inoue12}. 

The learning setting now is as follows. The positive examples are given by pairs of interpretations $(I,J)$ and the task is to find a program $P$ such that $T_P(I) = J$ holds for each such pair. Notice that by \prettyref{eq:T_P(I)=PI}, this is equivalent to finding a program $P$ satisfying $PI = J$ via sequential composition of $P$ with $I$.

We are thus interested in the transition behavior of programs with respect to the composition with interpretations and in this paper, we have derived numerous results in that direction. In fact, every result showing identities (modulo subsumption equivalence) between programs can be instantiated with interpretations to yield results concerning the behavior of programs with respect to interpretations. For example, \prettyref{t:(P_sqcup_R)Q} shows
\begin{align*} 
	(P\sqcup R)I = PI\sqcup RI \stackrel{\ref{eq:I_sqcup_J}}= PI\cap RI,
\end{align*} which means that if we have programs $P$ and $R$, then $P\sqcup R$ is a solution to the state transition $(I,PI\cap RI)$. Here the benefit of representing the solution program as the body-union of its two factors $P$ and $R$ is that it provides a better understanding of the inner workings of the program.






\subsection{Analogical logic programming}\label{sec:ALP}

In analogical logic programming --- recently introduced in \cite{Antic23-23} --- the task is as follows. 

Given a proportional equation of the form $P:Q::R:X$ --- read as ``$P$ is to $Q$ what $R$ is to $X$'' --- where $P,Q,R$ are programs and $X$ is a program variable, we wish to find solutions $S_1,S_2,\ldots$ such that $P:Q::R:S_i$ for all $i$, in which case we say that $P,Q,R,S_i$ are in an analogical proportion.\footnote{Analogical proportions between logic programs can be seen as analogical proportions between formulas \citep{Herzig24}.} That is, we generate novel programs from ``known'' ones via analogical transfer instead of starting with a partial specification as in ILP \prettyref{sec:ILP}. 

Analogical proportions formalize the idea that analogy-making is the task of transforming different objects from the source to the target domain in ``the same way'' (which is why ``copycat'' is the name of a prominent model of analogy-making \citep{Hofstadter95a}; and see \cite{Correa12}).

At the formal level, we instantiate a recently introduced abstract algebraic framework of analogical proportions \citep{Antic22} in our setting of logic programming, where we will restrict ourselves to the functorial fragment \citep{Antic23-23}:\footnote{The abstract framework has been studied in 2-valued boolean \citep{Antic21-3} and monounary algebras \citep{Antic22-2}, where analogical proportions admit particularly transparent structural descriptions.} 

\begin{definition}[Logic Program Proportions]\label{d:PQRS} Let $P,Q,R,S$ be programs. We say that \emph{$P$ transforms into $Q$ as $R$ transforms into $S$} --- in symbols, $P\to Q :: R\to S$ --- if there is a schema $\textbf{F}(X)$ such that $Q = \textbf{F}(P)$ and $S = \textbf{F}(R)$, in which case we call $\textbf{F}$ a \emph{justification} of $P\to Q :: R\to S$. We call an expression of the form $P\to Q :: R\to S$ an \emph{arrow proportion}. 

We now say that \emph{$P$ is to $Q$ as $R$ is to $S$} --- in symbols, $P:Q::R:S$ --- if $P\to Q :: R\to S$ and $Q\to P :: S\to R$, that is, if there are schemes $\textbf{F}$ and $\textbf{G}$ such that $Q = \textbf{F}(P)$, $S = \textbf{F}(R)$, $P = \textbf{G}(Q)$, and $R = \textbf{G}(S)$, in which case we call the pair $(\textbf{F},\textbf{G})$ a \emph{justification} of $P:Q::R:S$. We call an expression of the form $P:Q::R:S$ a \emph{(logic program) proportion}.
\end{definition}

Since \prettyref{d:PQRS} uses program schemes, the underlying algebra of programs is essential and, importantly, the set-like operations introduced in this paper allow us to derive (plausible) proportions which would otherwise not be expressible.

We shall now illustrate logic program proportions with some simple examples:

\begin{example} We have the proportion
\begin{align*} 
	\{a\} : \{a,b\} :: \{c\} : \{c,b\}
\end{align*} between interpretations as justified by the schemes
\begin{align*} 
	\textbf{F}(X) := X\cup \{b\} \quad\text{and}\quad \textbf{G}(X) := X- \{b\}.
\end{align*}
\end{example}

\begin{example} The proportion
\begin{align*} 
	\{a\} : \{a,b\} :: \{c\} : \{c\}
\end{align*} is justified by
\begin{align*} 
	\textbf{F}(X) := \left\{
	\begin{array}{l}
		a \leftarrow a\\
		b \leftarrow a\\
		c \leftarrow c
	\end{array}
	\right\}X \quad\text{and}\quad \textbf{G}(X) := \left\{
	\begin{array}{l}
		a \leftarrow a\\
		a \leftarrow b\\
		c \leftarrow c
	\end{array}
	\right\}X.
\end{align*}
\end{example}

\begin{proposition} For any programs $P$ and $R$,
\begin{align} 
	\label{eq:P_P^coyright_R_R^copyright}P : P^\copyright :: R : R^\copyright.
\end{align} 
\end{proposition}
\begin{proof} An immediate consequence of \prettyref{eq:P^copyright^copyright}.
\end{proof}

\begin{example} Let $A := \{a,b,c\}$. As a consequence of \prettyref{eq:P_P^coyright_R_R^copyright} we have, for example,
\begin{align} \label{eq:abc}
	\left\{
	\begin{array}{l}
		a \leftarrow b\\
		b \leftarrow a
	\end{array}
	\right\} : \left\{
	\begin{array}{l}
		a \leftarrow a,c\\
		b \leftarrow b,c
	\end{array}
	\right\} :: \left\{
	\begin{array}{l}
		a \leftarrow a\\
		b \leftarrow b\\
		c \leftarrow c
	\end{array}
	\right\} : \left\{
	\begin{array}{l}
		a \leftarrow b,c\\
		b \leftarrow a,c\\
		c \leftarrow a,b
	\end{array}
	\right\}
\end{align} which can be written more concisely as follows. 

Every permutation $\pi$ of $A$ gives rise to a \emph{permutation program} over $A$ --- again denoted by $\pi$ --- via
\begin{align} 
	\label{eq:pi} \pi := \{\pi(a) \leftarrow a \mid a\in A\}.
\end{align} We will use the well-known cycle notation.\footnote{\url{https://en.wikipedia.org/wiki/Permutation\#Cycle_notation}} For example, the permutation $\pi_{(a\,b\,c)}$, which is defined as the mapping $a\mapsto b$, $b\mapsto c$, $c\mapsto a$, gives rise to the permutation program
\begin{align*} 
	\pi_{(a\,b\,c)} = \left\{
	\begin{array}{l}
		b \leftarrow a\\
		c \leftarrow b\\
		a \leftarrow c
	\end{array}
	\right\}.
\end{align*} Notice that every permutation program is minimalist and Krom.

Moreover, define the \emph{diagonal program} over $A$ as the minimalist program
\begin{align} 
	\label{eq:D} \Delta_A := \{a \leftarrow (A- \{a\}) \mid a\in A\}.
\end{align} For example, we have
\begin{align*} 
	\Delta_{\{a\}} = \{a\},\quad \Delta_{\{a,b\}} = \pi_{(a\,b)},\quad \Delta_{\{a,b,c\}} = \left\{
	\begin{array}{l}
		a \leftarrow \quad b,c\\
		b \leftarrow a,\quad c\\
		c \leftarrow a,b
	\end{array}
	\right\}.
\end{align*}

We can now rewrite \prettyref{eq:abc} more compactly as
\begin{align*} 
	\pi_{(a\,b)} : \pi_{(a\,b)}^\copyright :: 1_A : \Delta_A,
\end{align*} where $\pi_{(a\,b)}$ is a permutation program, $1_A$ is a unit program defined in \prettyref{eq:1}, and $\Delta_A$ is a diagonal program. 
\end{example}

\begin{proposition}\label{p:P_P_sqcup_Q_} Let $P,Q,R$ be programs such that $Q$ is minimalist and
\begin{enumerate}[label=(\roman*)]
	\item\label{ite:B_cap_C} $B\cap C = \emptyset$, for all $a \leftarrow B\in P$ and $a \leftarrow C\in Q$;
	\item\label{ite:D_cap_C} $D\cap C = \emptyset$, for all $a \leftarrow D\in R$ and $a \leftarrow C\in Q$;
	\item\label{ite:h_P_cup_h_R} $h(P)\cup h(R)\subseteq h(Q)$.
\end{enumerate} Then
\begin{align*} 
	P : P\sqcup Q :: R : R\sqcup Q.
\end{align*}
\end{proposition}
\begin{proof} Since $Q$ is minimalist, for every $a\in h(Q)$ there is a unique rule $a \leftarrow C_a\in Q$. We show $(P\sqcup Q)\boxminus Q = P$; the identity $(R\sqcup Q)\boxminus Q = R$ is proved analogously using \ref{ite:D_cap_C} instead of \ref{ite:B_cap_C}. By \ref{ite:h_P_cup_h_R}, every rule head of $P$ occurs in $Q$, which by minimalism of $Q$ gives
\begin{align*} 
	P\sqcup Q = \{a \leftarrow (B\cup C_a) \mid a \leftarrow B\in P\}
\end{align*} and therefore, again by minimalism of $Q$,
\begin{align*} 
	(P\sqcup Q)\boxminus Q = \{a \leftarrow ((B\cup C_a) - C_a) \mid a \leftarrow B\in P\} = \{a \leftarrow (B - C_a) \mid a \leftarrow B\in P\} = P,
\end{align*} where the last identity holds by \ref{ite:B_cap_C}. Consequently, the schemes
\begin{align*} 
	\textbf{F}(X) := X\sqcup Q \quad\text{and}\quad \textbf{G}(X) := X\boxminus Q
\end{align*} satisfy $\textbf{F}(P) = P\sqcup Q$, $\textbf{F}(R) = R\sqcup Q$, $\textbf{G}(P\sqcup Q) = P$, and $\textbf{G}(R\sqcup Q) = R$, which by \prettyref{d:PQRS} means that $(\textbf{F},\textbf{G})$ is a justification of the stated proportion.
\end{proof}

\begin{example} As a consequence of \prettyref{p:P_P_sqcup_Q_}, we have
\begin{align*} 
	\{a\} : \{a \leftarrow c\} :: \{b\} : \{b \leftarrow c\},
\end{align*} justified by
\begin{align*} 
	\textbf{F}(X) := X\sqcup \left\{
	\begin{array}{l}
		a \leftarrow c\\
		b \leftarrow c
	\end{array}
	\right\} \quad\text{and}\quad \textbf{G}(X) := X\boxminus \left\{
	\begin{array}{l}
		a \leftarrow c\\
		b \leftarrow c
	\end{array}
	\right\}.
\end{align*} 
\end{example}






\section{Conclusion}

In this paper, we introduced a family of algebraic operations on propositional Horn logic programs that generalize classical set-theoretic constructions, including union, intersection, complement, subtraction, symmetric difference, and power set (see \prettyref{sec:Body-union}--\prettyref{sec:Body-power-set}). While inspired by set theory, these operations exhibit fundamentally different behavior due to the rule-based structure of logic programs, in particular through non-idempotence and non-trivial interactions between rules sharing the same head.

These operations provide a structured means of decomposing and recomposing logic programs, which allowed us to establish the Decomposition Theorems \ref{t:MDT} and \ref{t:DT}. As illustrated in \prettyref{sec:App}, they offer a basis for approaching problems in automatic logic programming from an algebraic perspective, including program construction, transformation, and analysis.

A natural direction for future work is the extension of this framework beyond the propositional Horn setting. In particular, lifting these operations to first-order and higher-order logic programs \citep{Apt90,Chen93,Lloyd87,Miller12}, as well as to answer set programs with negation as failure \citep{Gelfond91,Marek99,Lifschitz19,Brewka11}, raises substantial technical challenges. For instance, the definition of a body-complement operation in the first-order setting is non-trivial, as rule bodies are required to be finite, preventing a direct generalization of the propositional case. Developing suitable generalizations that preserve the algebraic structure identified in this paper remains an open problem.

More broadly, the results highlight an algebraic structure of the class of logic programs considered here: minimalist programs can be decomposed into and constructed from Krom components while preserving their least-model semantics, whereas arbitrary programs admit corresponding syntactic and semantic approximations. In this respect, the results complement earlier algebraic approaches to logic programming and may support the further development of compositional methods for program synthesis and transformation.

\section*{Acknowledgments}

I would like to thank the reviewers for their thoughtful and constructive comments which improved the quality of the article.

\section*{Use of artificial intelligence tools}

OpenAI Codex (GPT-5, accessed on 4 August 2026, \url{https://openai.com/codex/}) was used to analyse an earlier proof of the Minimalist Decomposition \prettyref{t:MDT}, identify a gap in that argument, and assist in drafting a corrected proof.

Anthropic's Claude (Opus 5, accessed on 7 August 2026, \url{https://claude.ai}) was subsequently used to check the manuscript for mathematical and expository errors. It assisted in testing algebraic identities, identifying counterexamples, missing hypotheses, and an omitted empty-body case, and in drafting and verifying corrected statements, proofs, and examples. This included the results on distributivity and subsumption equivalence in \prettyref{t:(P_sqcup_R)Q}, \prettyref{t:sqcup_sqcap}, and \prettyref{l:sqcup_congruence}, as well as the Decomposition Theorems \ref{t:MDT} and \ref{t:DT}.

Both systems also assisted in revising the abstract and the introduction. The author reviewed and independently verified all resulting statements and arguments and takes full responsibility for the correctness and content of the manuscript.

\section*{Competing interests}

The author declares none.

\section*{Data availability statement}

The manuscript has no data associated.

\bibliographystyle{tlplike}
\bibliography{/Users/christianantic/Bibdesk/Bibliography,/Users/christianantic/Bibdesk/Publications_J,/Users/christianantic/Bibdesk/Publications_C,/Users/christianantic/Bibdesk/Publications_arXiv,/Users/christianantic/Bibdesk/Submitted}
\if\isdraft1

\newpage
\section*{TODOs}

\todo[inline]{\cite{Antic26-1} noch hinzunehmen}

\fi
\end{document}